\documentclass[a4paper,twocolumn,superscriptaddress,11pt]{quantumarticle}
\pdfoutput=1
\usepackage[utf8]{inputenc}
\usepackage[english]{babel}
\usepackage[T1]{fontenc}
\usepackage{amsthm}
\usepackage{hyperref}
\usepackage{dsfont}
\usepackage{appendix}

\usepackage{amsmath}
\usepackage{amssymb}

\newcommand{\R}{\mathbb{R}}

\usepackage{tikz}
\usepackage{lipsum}

\definecolor{rev1}{rgb}{0,0,1}

\newcommand{\revv}[1]{#1} 
\newcommand{\rev}[1]{#1}

\newcommand{\snote}[1]{} 

\usepackage[normalem]{ulem}

\newtheorem{thm}{Theorem}
\newtheorem{prop}{Proposition}
\newtheorem{deff}{Definition}

\newtheorem{expl}{Example}

\newcommand{\poly}{\mbox{{\rm poly}}}

\newcommand{\PP}{\mbox{\rmfamily\textsc{P}}} 
\newcommand{\FP}{\mbox{\rmfamily\textsc{FP}}} 
\newcommand{\NP}{\mbox{\rmfamily\textsc{NP}}} 
\newcommand{\coNP}{\mbox{\rmfamily\textsc{coNP}}} 
\newcommand{\FNP}{\mbox{\rmfamily\textsc{FNP}}} 
\newcommand{\TFNP}{\mbox{\rmfamily\textsc{TFNP}}}

\newcommand{\BQP}{\mbox{\rmfamily\textsc{BQP}}} 
\newcommand{\coBQP}{\mbox{\rmfamily\textsc{coBQP}}} 
\newcommand{\FBQP}{\mbox{\rmfamily\textsc{FBQP}}} 
\newcommand{\QMA}{\mbox{\rmfamily\textsc{QMA}}} 
\newcommand{\coQMA}{\mbox{\rmfamily\textsc{coQMA}}} 
\newcommand{\FQMA}{\mbox{\rmfamily\textsc{FQMA}}} 
\newcommand{\TFQMA}{\mbox{\rmfamily\textsc{TFQMA}}} 

\newcommand{\LTHREE}{\mbox{\rmfamily\textsc{L3}}} 
\newcommand{\LTWO}{\mbox{\rmfamily\textsc{L2}}}

\newcommand{\MA}{\rm MA}
\newcommand{\coMA}{\rm coMA}
\newcommand{\FMA}{\rm FMA}
\newcommand{\TFMA}{\rm TFMA}

\newcommand{\BPP}{\rm BPP}
\newcommand{\FBPP}{\rm FBPP}

\newcommand{\HH}{\mathcal{H}}

\newcommand{\suppress}[1]{}

\newcommand{\SPECT}{\mbox{\rmfamily\rm{Spect}}}
\newcommand{\SPAN}{\mbox{\rmfamily\rm{Span}}}

\def\N{{\mathbb N}}

\def\Tr{\mathrm{Tr}}

\begin{document}

\title{Characterizing the intersection of QMA and coQMA.}
\date{\today}

\author{Serge Massar}
\affiliation{Laboratoire d'Information Quantique CP224, Universit\'e libre de Bruxelles, B-1050 Brussels, Belgium.}
\email{smassar@ulb.ac.be}
\orcid{0000-0002-4381-2485}
\author{Miklos Santha}
\affiliation{CNRS, IRIF, Universit\'e Paris Diderot, 75205 Paris, France.}
\affiliation{Centre for Quantum Technologies \& MajuLab, National
University of Singapore, Singapore.}

\maketitle

\begin{abstract}
We show that the functional analogue of $\QMA \cap \coQMA$, denoted ${\rm F}(\QMA \cap \coQMA)$, equals the complexity class  Total Functional $\QMA$ ($\TFQMA$). To prove this we need to introduce an alternative definition of $\QMA \cap \coQMA$ in terms of a single quantum verification procedure.
We  show that if $\TFQMA$ equals the functional analogue of $\BQP$ ($\FBQP$), then $\QMA \cap \coQMA = \BQP$. We show that if there is a $\QMA$ complete problem that (robustly) reduces to a problem in $\TFQMA$, then $\QMA \cap \coQMA = \QMA$.  Our results thus imply that if some of the inclusions between functional classes
$\FBQP\subseteq \TFQMA \subseteq \FQMA$ are in fact equalities,  then the corresponding inclusions in
$\BQP\subseteq \QMA \cap \coQMA \subseteq \QMA$ are also equalities.
\end{abstract}

\section{Introduction}\label{Intro}

Functional  $\NP$ ($\FNP$) is 
the class of search problems defined by polynomial time relations $M(x,y)$ (the verifier) where the length of
$y$ is polynomial in the length of $x$. On an input $x$ the task is to find a witness $y$ (if it exists) such that 
$M$ accepts $(x,y)$.

The complexity class Total Functional $\NP$ ($\TFNP$), introduced in \cite{MP91}, is the subset of $\FNP$ for which it can be shown that for all inputs $x$, there exists at least one witness $y$.
It lies between Functional $\PP$ ($\FP$) 
(the subclass of $\FNP$ for which a witness can be found in polynomial time), and $\FNP$.

$\TFNP$ contains many natural and important problems, including factoring, local search problems\cite{JPY88,PSY90,Kr89}, computational versions of Brouwer's fixed point theorem\cite{P94}, finding Nash equilibria\cite{DGP09,CDT09}. Although there probably do not exist complete problems for $\TFNP$,  there are many syntactically defined subclasses of $\TFNP$ that contain complete problems, and for which some of the above natural problems can be shown to be complete. 
For recent work in this direction, see \cite{GP17}.

One of the results of the founding paper \cite{MP91} concerns the relation between 
Total Functional $\NP$ and  other complexity classes. The inclusions
\begin{equation}
\FP \subseteq \TFNP \subseteq \FNP
\label{Eq:SubseteqTFNP}
\end{equation}
are obvious. But what would be the consequences if some of these inclusions were not strict, but replaced by equality?

In \cite{MP91} this question is connected with the inclusions
\begin{equation}
\PP \subseteq \NP \cap \coNP \subseteq \NP\ .
\label{Eq:SubseteqNPcoNP}
\end{equation}
It  was first shown in~\cite{MP91} that Total Functional $\NP$ equals the functional analogue of $\NP\cap \coNP$ denoted ${\rm F}(\NP \cap \coNP)$:
\begin{equation}
\TFNP = {\rm F}(\NP \cap \coNP)\ ;
\label{Eq:TFNP-basic}
\end{equation}

Furthermore,  it is proven in~\cite{MP91} that an inclusion in Eq. \eqref{Eq:SubseteqTFNP}
is strict if and only if the corresponding inclusion in Eq. \eqref{Eq:SubseteqNPcoNP} is strict.
\rev{ The  inclusions in Eq. \eqref{Eq:SubseteqNPcoNP}  are believed to be strict,
as there are problems, such as factoring, that belong to $\NP\cap \coNP$ but 
 are believed to not belong to $\PP$, and there are problems such as 3-SAT that belong to $\NP$ but are thought not to belong to $\coNP$.
This result therefore provides strong evidence }  that the 
 inclusions in Eq. \eqref{Eq:SubseteqTFNP} are strict. The proofs of these results are very simple, and take 
 only one paragraph, or are even just implicit in \cite{MP91}.

The quantum analogue of $\NP$ is $\QMA$ \cite{KSV02}. $\QMA$ has been extensively studied, and contains a rich set of complete problems, see e.g. \cite{B12}.
Functional $\QMA$, the problem of producing a quantum state that serves as witness for a $\QMA$ problem was first introduced in \cite{JWB03}.

In a recent work  \cite{MS18} 
we 
 introduced  the complexity class Total Functional $\QMA$ ($\TFQMA$), the subset of $\FQMA$ for which it can be shown that for all inputs $x$, there exists at least one witness $\vert \psi \rangle$. $\TFQMA$ is an expressive class containing several interesting problems.
However in~ \cite{MS18}
the analog of the complexity results of \cite{MP91} for $\TFQMA$ where left open. Here we prove 
these results.

The analogs of Eqs. \eqref{Eq:SubseteqTFNP} and \eqref{Eq:SubseteqNPcoNP} are
\begin{equation}
\FBQP \subseteq \TFQMA \subseteq \FQMA
\label{Eq:SubseteqTFQMA}
\end{equation}
and
\begin{equation}
\BQP \subseteq \QMA \cap \coQMA \subseteq \QMA\ .
\label{Eq:SubseteqQMAcoQMA}
\end{equation}

We first show that $\TFQMA$  equals the
functional analogue of $\QMA \cap \coQMA$:
\begin{equation}
\TFQMA = {\rm F}(\QMA \cap \coQMA)\ .
\label{Eq:TFQMA-basic}
\end{equation}

We then show that if $\FBQP = \TFQMA$, then $\BQP = \QMA \cap \coQMA$.
We finally show that if there is a $\QMA$ complete problem that reduces (using a slightly stronger notion of reduction than the natural one, which we call robust reduction) to a problem in
$\TFQMA$, then $\QMA \cap \coQMA = \QMA$.

But while the proofs in the classical case are elementary, the
quantum proofs are more delicate.

 Showing  that $\TFQMA ={\rm F}(\QMA \cap \coQMA)$ requires a detailed enquiry into what is the correct  definition of  ${\rm F}(\QMA \cap \coQMA)$.
The  difficulty is that 
any language $L$ in $\QMA \cap \coQMA$ is naturally defined by  two quantum verification procedures $Q$ and $Q'$.  But these two quantum verification procedures do not necessarily commute.  Therefore it is not clear how to define a witness for $\QMA \cap \coQMA$. Indeed  given an input $x$ and a state $\vert \psi\rangle$ one cannot test  whether both $Q(x,\vert \psi \rangle)$ and $Q'(x,\vert \psi \rangle)$ accept, since the act of carrying out one of these tests will modify the state  $\vert \psi\rangle$ and render impossible the other test.

We will see that the solution to this conundrum is to append to the state a bit $z\in \{0,1\}$, with the convention that if $z=0$ one only tests $Q$, and if $z=1$ one only tests $Q'$. Thus a witness for $\QMA \cap \coQMA$ has the form $\vert z \rangle \vert \psi\rangle$.

However, while this solution is natural, it is not clear whether it is the unique way to solve this problem. In order to address this, we need a deeper understanding of $\QMA \cap \coQMA$. 
To this end we introduce two alternative definitions of  $\QMA \cap \coQMA$. The more
important is a definition in terms of a single 2-outcome quantum verification procedure which takes as input a single witness.
This  definition is particularly useful because it allows us to apply to  $\QMA \cap \coQMA$ the notions of of eigenbasis, spectrum, and eigenspace of a quantum verification procedure which we introduced in \cite{MS18} (based on the earlier work of Marriott and Watrous\cite{MW05}).
These notions are then used to provide a natural definition for ${\rm F}(\QMA \cap \coQMA)$ which parallels the definition for $\FQMA$ given in  \cite{MS18}. 

Furthermore, by extending the 
notion of QMA amplification as developed in \cite{MW05}, we  show that our definition of ${\rm F}(\QMA \cap \coQMA)$ is independent of the completeness and soundness bounds.

Together these results allow us to define precisely the right hand side of Eq. \eqref{Eq:TFQMA-basic}, in such a way that it has the same structure as the left hand side of Eq. \eqref{Eq:TFQMA-basic}. Proving the equality of the two quantities is then rather easy.

The paper is structured as follows:

Sections \ref{Sec:PrelDef} to \ref{Sec:FQMA}
sets the stage, defining quantum verification procedures, $\QMA$, $\coQMA$, $\QMA \cap \coQMA$,  Functional $\QMA$, Total Functional $\QMA$, and recalling the key notions of eigenbasis, spectrum and eigenspaces of a quantum verification procedure.
It closely follows, with appropriate modifications, our recent work \cite{MS18} on $\TFQMA$.

Section \ref{Sec:Reductions} introduces the notion of reduction of quantum verification procedures, including the notion of robust reduction mentioned above.

Section \ref{Sec:E-maps} recalls the notion of eigenspace preserving map introduced in \cite{MS18}.

Section \ref{Sec:GenAmp} introduces ``iterative procedures'', thereby generalising an idea introduced in \cite{MW05}. Iterative procedures modify the acceptance probabilities of eigenstates of a quantum verification procedure without changing the eigenstates themselves. We use iterative procedures to show \rev{ in Theorem \ref{Thm:SuperStrongAmpl}} that the spectrum of a quantum verification procedure can  be modified 
\revv{to a very large extent.}

As a first application of iterative procedures  we introduce in Section \ref{Sec:NonDest} the notion of ``non-destructive'' procedure which outputs both a classical bit indicating whether the procedure accepts or rejects and a quantum state, in such a way that if the input state is an eigenstate, the output state is also the eigenstate. We show that, essentially without loss of generality, one can take quantum verification procedures to be non-destructive.

In Section \ref{Sec:EquivQMAcapcoQMA} we present the two additional definitions of $\QMA \cap \coQMA$ mentioned above.
We show that the three definitions are equivalent. We then introduce in Section \ref{Sec:FQMAcoQMA} functional $\QMA \cap \coQMA$ and use 
the  results of Section  \ref{Sec:GenAmp} to show that this definition does not depend on the  soundness and  completeness thresholds.

In Section \ref{Sec:TFQMAequal} we show that 
${\rm F}(\QMA \cap \coQMA)$ 
equals the class Total Functional $\QMA$ ($\TFQMA$).

Section \ref{Sec:InclFBQP} contains the proof that if $\FBQP = \TFQMA$, then $\BQP = \QMA \cap \coQMA$.

And finally in Section \ref{Sec:RobReducQMATFQMA} we prove that if there exists a $\QMA$ complete problem that robustly reduces to a problem in $\TFQMA$, then  $\QMA = \QMA \cap \coQMA $.

\rev{
In the appendix, we sketch how similar results hold for the classical probabilistic classes $\MA$ and $\coMA$. We recommend that the reader interested in the details of the proofs of the paper read this appendix simultaneously with the rest of the paper, as it explains  which parts of our proofs are quantum and which parts are in fact classical.} 

\section{Preliminary Definitions}\label{Sec:PrelDef}

We denote by $\HH_n$  the Hilbert space of $n$ qubits. 
For pure states we use the Dirac ket notation  $\vert \psi \rangle$, whereas for density matrices we just use the Greek letter $\rho$. We denote by $\vert 0^k \rangle$ the state of $k$ qubits all in the $\vert 0\rangle$ state. We denote by $I_m$ the identity operator acting on $m$ qubits.

We denote by $\poly$ the set of all non-zero polynomials with non-negative integer coefficients. Note that if $f\in \poly$, then $f$ maps positive integers to positive integers.

We denote by $1/\poly$ the set of all functions that are the inverse of a polynomial in $\poly$:
\begin{equation}
1/\poly=\{g:\N \to \R \ : \exists p\in \poly \textrm{ AND } g=1/p\}\ .
\label{Eq:1/poly}
\end{equation}

\section{Quantum Verification Procedures}
\label{Sec:QVP}

\begin{deff}{\bf $d$-Outcome Quantum Verification Procedure.} 
\label{def:d-qvp}
For any integer $d$ larger or equal to $2$, a {\em $d$-outcome quantum verification procedure} is a 
polynomial time 
uniform family of quantum circuits $Q=\{Q_{n} : n \in \N\}$ 
with
$Q_n$ taking as input $(x, \vert \psi \rangle\otimes \vert 0^k \rangle )$, 
where 
 $x \in \{0,1\}^n$ is a binary string of length $n$,
$\vert \psi \rangle$ is a state of 
$m$ qubits, and both $m= m(n)$ and $k= k(n)$ belong to $\poly$.
The last $k$ qubits, initialized to the state $\vert 0^k \rangle$, form the {\em ancilla Hilbert space} $\HH_k$, 
and the $m$-qubit states $\vert \psi \rangle$ form the {\em witness Hilbert space} $\HH_m$.
The outcome of the run of $Q_n$ is a  word  $w\in\{0,\ldots d-1\}$.
It  is obtained by measuring the first $\lceil \log d \rceil$ qubits 
in the computational basis, interpreting the result as an integer in $\{0,...,2^{\lceil \log d \rceil}-1\}$ 
and, if  this integer is greater or equal to $d$, replacing it by $d-1$.
We denote this outcome by $Q_{n}(x,\vert \psi \rangle)$.
\end{deff}

In most of this article we will consider 2-outcome quantum verification procedures. For brevity throughout this work we use the following terminology which is standard in the literature:
\begin{deff}{\bf Quantum Verification Procedure.} 
A 2-outcome quantum verification procedure is called a {\em quantum verification procedure.}
\label{def:Qvp}
\end{deff}

Note that a $d$-outcome quantum verification procedure can of course also take as input a mixed state $\rho$, rather than a pure state  $\vert \psi \rangle$. Mixed states can be written as convex combinations of pure states. The acceptance (rejection) probability for the mixed state is the convex combination of the acceptance (rejection) probabilities for the constituent pure states.

In order to simplify notation, in what follows we will mainly consider the case where $Q_{n}$ takes as input a pure state. In view of the above remark, the extension to mixed state inputs is immediate. In some cases the argument requires that $Q_{n}$ takes as input a mixed state, in which case,
abusing slightly the notation, we write $Q_{n}(x,\rho)$ for the outcome of the quantum verification procedure on the mixed state $\rho$.

\section{$\QMA$}
\label{Sec:QMA}

\begin{deff}{\bf (a,b)--Quantum Verification Procedure.} 
\label{def:qvpab}
Let
$a,b : \N \rightarrow (0,1)$ be polynomially time computable functions
 which satisfy 
\begin{equation}
a(n)-b(n)\geq 1/ q(n)\ ,
\end{equation}
for some $q\in \poly$.
We say that a quantum verification procedure $Q$ is an 
$(a,b)$-{\em quantum verification  procedure} 
(or shortly an  $(a,b)$-{\em procedure})
if for every $x$ of length $n$,  one of the following holds:
\begin{eqnarray}
\exists \vert \psi \rangle,  \Pr [Q_{n}(x,\vert \psi \rangle)=1]\geq a,\ \label{QMA1}\\
\forall \vert \psi \rangle, \Pr [Q_{n}(x,\vert \psi \rangle )=1]\leq b.\ \label{QMA2}
\end{eqnarray}
We call $a$ and $b$  the {\em completeness} and {\em soundness} probabilities of the quantum verification procedure.
\end{deff}

Note that taking the completeness or soundness probabilities equal $1$ or $0$ require a specific formulation dealing with exact quantum computation, which while theoretically interesting is not implementable in practice (real quantum computation will have some, possibly exponentially small, error probability). One of the results of the present work is to show that several complexity classes do not depend on the completeness and soundness probabilities used to define them. More precisely, in these results we show that the completeness and soundness probabilities can be taken exponentially close to $1$ and  $0$ respectively. These results do not extend to showing that $a$ and $b$ can be taken equal to $1$ and $0$.  For these reason we explicitly  exclude in Definition \ref{def:qvpab} and all similar definitions  the cases when $a=1$ and $b=0$, and take them to belong to the open interval $(0,1)$.

\begin{deff}
\label{Deff:QMA_and_coQMA}
{\bf QMA and coQMA.} Let $a,b$ be functions as in Definition$~\ref{def:qvpab}$.
The {\em class} $\QMA(a,b)$ is the set of languages $L \subseteq\{0,1\}^*$ such that there exists an
$(a,b)$-procedure $Q$, where for every x, we have
$x \in L$ if and only if 
Equation~\eqref{QMA1} holds
(and consequently, $x \notin L$ if and only if 
Equation~\eqref{QMA2} holds).

We call $Q$ a 
{\em quantum verification procedure for $L$}, and 
for $x\in L$, we say that a $\vert \psi \rangle$ satisfying Equation~\eqref{QMA1}
is a {\em witness} for $x$.

The {\em class $\coQMA(a,b)$} is the set of languages $L \subseteq \{0,1\}^*$ such that there exists an
$(a,b)$-quantum verification procedure $Q'$, where for every x, we have
$x \notin L$ if and only if 
Equation~\eqref{QMA1} holds 
(and consequently, $x \in L$ if and only if Equation~\eqref{QMA2} holds).
\end{deff}

\begin{deff}
\label{Deff:QMACAPcoQMA}
{\bf QMA $\cap$ coQMA.} Let $a,b$ and $a',b'$ be pairs of functions as in Definition$~\ref{def:qvpab}$.
The {\em class} $\QMA\cap \coQMA(a,b;a',b')$ is the set of languages $L \subseteq\{0,1\}^*$ such that
$L\in \QMA(a,b)$ and $L\in \coQMA(a',b')$.

More explicitly, $\QMA\cap\coQMA$ is the class of languages  $L\subseteq \{0,1\}^*$ such that there   exist an $(a,b)$-procedure $Q=\{Q_n\}$, and an $(a',b')$-procedure $Q'=\{Q'_n\}$, such that for every $x$ we have $x\in L$ if and only if  both the following hold:
\begin{eqnarray}
\exists \vert \psi \rangle,  \Pr [Q_{n}(x,\vert \psi \rangle)=1]\geq a,
\label{QMAcoQMA1}\\
\forall \vert \psi' \rangle, \Pr [Q'_{n}(x,\vert \psi'\rangle )=1]\leq b';
\label{QMAcoQMA2}
\end{eqnarray}
and $x\notin L$ if and only if  both the following hold:
\begin{eqnarray}
\forall \vert \psi \rangle,  \Pr [Q_{n}(x,\vert \psi \rangle)=1]\leq b,
\label{QMAcoQMA3}\\
\exists \vert \psi' \rangle, \Pr [Q'_{n}(x,\vert \psi' \rangle )=1]\geq a'.
\label{QMAcoQMA4}
\end{eqnarray}

\end{deff}

It is of course essential to understand to what extent the above definitions depend on the bounds $(a,b)$ and $(a',b')$.
It was shown by Kitaev  that the separation $a-b$ in the definition of $\QMA$ could be amplified exponentially
by using multiple copies of the input state, and multiple copies of the verification circuit \cite{KSV02}, that is
by increasing both $m$ and $k$. This was further improved in 
\cite{MW05} (see also \cite{NWZ09})  where it was shown that by running forwards and backwards the original quantum verification procedure, only one copy of the input state was needed to obtain the same amplification, 
that is one needs only increase $k$. 

\begin{thm}
{\bf QMA Amplification \cite{KSV02,MW05,NWZ09}.} 
\label{Thm-QMA-Amplification}
Let $a,b$ be functions as in Definition$~\ref{def:qvpab}$. 
 For any
polynomial $r$, we have
$\QMA(a,b) \subseteq \QMA(1-2^{-r},2^{-r})$. 
\end{thm}

As a consequence of Theorem \ref{Thm-QMA-Amplification} the precise values of the bounds $(a,b)$  are irrelevant. Traditionally they are taken to be $2/3$ and $1/3$. 
We will do here the same.

\begin{deff}
\label{Deff:QMA}
We define the class $\QMA$ as $\QMA(2/3, 1/3).$
\end{deff}

\begin{deff}
\label{Deff:QMACAPCOQMA}
We define the class $\QMA \cap \coQMA$ as $\QMA\cap \coQMA(2/3,1/3;2/3,1/3)$
\end{deff}

\begin{deff}{\bf $a$-Total Quantum Verification Procedure.} 
\label{Deff:Totalqvp}
Let $a : \N \rightarrow [0,1]$ be a polynomially time computable function.
We say that a quantum verification procedure is an 
$a$-{\em total quantum verification  procedure} (or shortly an 
$a$-{\em total procedure})
if for every $x$ of length $n$,  the following holds:
\begin{equation}
\exists \vert \psi \rangle,  \Pr [Q_{n}(x,\vert \psi \rangle)=1]\geq a \ . \label{TQMA1}
\end{equation}
\end{deff}

Note that an $a$-total procedure is also an $(a,b)$-procedure for all $b$ satisfying the conditions of Definition \ref{def:qvpab}. 
 Note  that the language associated to an $a$-total procedure is $L=\{0,1\}^*$. That is the decision problem for total procedures is trivial, since for all $x \in \{0,1\}^*$ there exists a witness for $x$. Therefore for total procedures, the only interesting questions concern the witnesses. 
 The main topic of the present work is to understand the relation between $a$-total procedures  and $\QMA \cap \coQMA$.

 \section{BQP}
 \label{Sec:BQP}
 
\subsection{Efficiently implementable channels and states.}

Recall that a quantum channel is a completely positive (CP) trace preserving map between spaces of operators.

\begin{deff}{\bf Efficiently implementable quantum channel.} 
\label{def:EffChannel}

A family $\Phi$  of efficiently implementable quantum channels
is defined by a 
polynomial time 
uniform family of quantum circuits $Q=\{Q_n : n \in \N\}$ 
with 
$Q_n$ taking as input $(x,\vert \psi \rangle\otimes \vert 0^{k} \rangle )$
where
$x\in \{0,1\}^n$, where $\vert \psi \rangle \in \HH_{m}$ \rev{where $m=m(n)$ and $k=k(n)$ belong to $\poly$}, 
and  the output of the channel is
obtained by keeping the first $m'$ qubits and tracing over the remaining $m+k-m'$ qubits, where \rev{where $m'=m'(n)$  belongs to $\poly$}  and $m'\leq m+k$. 
We denote the output of the channel by $\Phi(x,\vert \psi \rangle)$:
\begin{equation}
\Phi(x,\vert \psi \rangle) = \Tr_{m+k-m'} Q_n(x,\vert \psi \rangle\otimes \vert 0^{k} \rangle )\ .
\end{equation}
\end{deff}

Note that a quantum channel can of course also act on a mixed state $\rho$ of $m$ qubits.

\begin{deff}
\label{Deff:EffState}
{\bf Efficiently preparable states.} 
Let $m' \in \poly$.
A family of density matrices $\{ \rho(x) : x\in\{0,1\}^n,  n \in \N\}$
is efficiently preparable if 
$\rho(x)$ acts on $\HH_{m'(n)}$ and if
there exists 
a 
polynomial time 
uniform family of quantum circuits $Q=\{Q_{n} : n \in \N\}$ 
with
$Q_n$ taking as input $(x,  \vert 0^{k} \rangle )$
with 
 $x \in \{0,1\}^n$, \rev{with $k=k(n)$ belonging to $\poly$ 
and   $k \geq m'$}, 
and where $\rho(x)$ is obtained by tracing out the last $k-m'$ qubits of $Q_n(x)$.
\end{deff}

This definition is equivalent to saying that a family of efficiently preparable states $\{\rho(x)\}$ is the output of a family of efficiently implementable quantum channels $\Phi$ 
which does not take any quantum state as input, i.e. if one sets $m=0$ in Definition \ref{def:EffChannel}: $\rho(x)=\Phi(x,\emptyset)$.

 \subsection{$\BQP$}
 
 Bounded-error quantum polynomial time ($\BQP$) is the class of decision problems solvable by a quantum computer in polynomial time, with bounded error probability for all instances. We give two equivalent definitions of $\BQP$, following the formulation of \cite{MS18}.

\begin{deff}{\bf $\BQP$.} 
\label{def:BQP}
Let $a,b$  be  functions as in Definition$~\ref{def:qvpab}$.
The class $\BQP(a,b)$ is the set of languages $L\subseteq \{0,1\}^*$ such that there exists an $(a,b)$-procedure $Q=\{Q_{n} : n \in \N\}$, 
with
$Q_n$ taking as input $(x,  \vert 0^k \rangle )$ (i.e. there is no witness Hilbert space), where
 $x \in \{0,1\}^n$ is a binary string of length $n$, and where for every $x$  we have
\begin{eqnarray}
x\in L &\Leftrightarrow & 
\Pr [Q_{n}(x)=1]\geq a\ ,\ \label{BQP1}\\
x\notin L &\Leftrightarrow & 
\Pr [Q_{n}(x)=1]\leq b\ .
\ \label{BQP2}
\end{eqnarray}
\end{deff}

 By repeating the procedure a polynomial number of times, the thresholds can be made exponentially close to $1$ and $0$ respectively. Therefore the exact values of the bounds $a$ and $b$ are irrelevant. Traditionally they are taken to be $2/3$ and $1/3$. 
We will do here the same.
 
 \begin{deff}
\label{Deff:BQP}
We define the class $\BQP$ as $\BQP(2/3, 1/3).$
\end{deff}

We  give an alternative definition of $\BQP$ which is closer to the definition of $\QMA$.

\begin{deff}
\label{Deff:Lpoly}
{\bf The language class $\BQP'$.} 
\label{Def:BQP'}

Let $a,b$  be  functions as in Definition$~\ref{def:qvpab}$. Let $\BQP'(a,b)\subseteq \QMA(a,b)$ be 
the set of languages  $L\subseteq \{0,1\}^*$ such that
there   exists:
\begin{enumerate}
\item  an $(a,b)$ quantum verification procedure $Q=\{Q_{n} : n \in \N\}$ with
$Q_n$ taking as input $(x, \vert \psi \rangle\otimes \vert 0^k \rangle )$, 
where 
 $x \in \{0,1\}^n$ is a binary string of length $n$,
$\vert \psi \rangle$ is a state of
$m$ qubits, with $m,k\in \poly$;
\item an efficiently preparable set of density matrices $\{\rho(x)\}$ where $\rho(x)$ acts on $\HH_m$;
\end{enumerate}
and where for every x, we have
 $x \in L$ if and only if 
\begin{eqnarray}
\Pr [Q_{n}(x,\rho(x))=1]\geq a,\ \label{BQP3}
\end{eqnarray}
and
 $x \notin L$ if and only if 
\begin{eqnarray}
\forall \vert \psi \rangle, \Pr [Q_{n}(x,\vert \psi \rangle )=1]\leq b.\ \label{QMA2B}
\end{eqnarray} 
 \end{deff}

\begin{thm}
{\bf \BQP'=\BQP\cite{MS18}.} 
\label{Thm:LPOLY=BQP}
For all $a,b$    functions as in Definition$~\ref{def:qvpab}$, $\BQP'(a,b)=\BQP$.
\end{thm}

\section{Structure of the witness space.}
\label{SubSec:AmplStructureWitness}

The witness space has a structure that will be central in what follows. This structure can be derived using the methods of
\cite{MW05,NWZ09} based on Jordan's lemma, or more easily using the structure of the POVM element corresponding to the quantum verification procedure accepting, see \cite{ABBS08}.
We repeat here the formulation of  \cite{MS18}.

\begin{thm}
{\bf Structure of witness space \cite{MW05,ABBS08}.} 
\label{Thm:BlockStructure}
Given a  
quantum verification procedure $Q=\{Q_{n}\}$, for all $x \in\{0,1\}^n$
there exists a basis $B_Q(x)=\{\vert \psi_i\rangle : 1 \leq i \leq 2^m\}$ of the witness space  $\HH_m$
such that
the acceptance probability of linear combinations of the basis states does not involve interferences,
that is for all $\alpha_i$ such that $\sum_i \vert\alpha_i\vert^2=1$, we have
\begin{eqnarray}
 &\Pr [Q_{n}(x,\sum_i \alpha_i \vert \psi_i \rangle) = 1 ]&
\nonumber\\ 
& =
\sum_i \vert\alpha_i\vert^2  \Pr [Q_{n}(x,\vert \psi_i \rangle)
=1]&\ .
\label{Eq:eigenbasis}
\end{eqnarray}
\end{thm}

\begin{deff}
\label{Deff:eigenbasis_spectrum_qvp}
{\bf Eigenbasis, spectrum and eigenspaces of a quantum verification procedure.}
Given a  
quantum verification procedure $Q=\{Q_{n}\}$, for all $x \in\{0,1\}^n$ and for any 
eigenbasis $B_Q(x)=\{\vert \psi_i\rangle \}$ of $Q$ for $x$,
\begin{enumerate}
\item for all $\vert \psi_i\rangle\in B_Q(x)$, we call 
\begin{equation}
p_i= \Pr [Q_{n}(x,\vert \psi_i \rangle)
=1]
\end{equation}
 the \emph{acceptance probability} of $\vert \psi_i\rangle$;
\item 
we call the set of acceptance probabilities the \emph{spectrum} of $Q$  for $x$:
\begin{eqnarray}
\label{Eq:SPECT(Q,x)}
&\SPECT(Q,x) =\{p\in[0,1] :  \exists \vert \psi_i\rangle \in B_Q(x) &\nonumber\\
&\quad  {\rm such~that\ } 
\Pr [Q_{n}(x,\vert \psi_i \rangle) = 1 ]=p\}\ ;&\nonumber\\
\end{eqnarray}
\item for  $p\in\SPECT(Q,x)$, we call 
\begin{eqnarray}
\label{Eq:HH(Q,x)}
\HH_Q(x,p)&=&\SPAN ( 
\{ \vert \psi_i \rangle \in B_Q(x)  \nonumber\\
& &\ :\  \Pr [Q_{n}(x,\vert \psi_i \rangle) = 1 ]=p\} )\nonumber\\
\end{eqnarray}
the \emph{eigenspace} of $Q$ for $x$ with acceptance probability $p$.
\end{enumerate}
\end{deff}

\begin{thm}
{\bf Uniqueness of the spectrum and eigenspaces of $Q$ \cite{ABBS08,MS18}.} 
\label{Thm:UniquenessBlockStructure}
Given a  
quantum verification procedure $Q=\{Q_{n}\}$, for all $x \in\{0,1\}^*$,
the spectrum  $\SPECT(Q,x)$ of $Q$ 
and
the eigenspaces $\HH_Q(x,p)$ of $Q$ with acceptance probability $p\in \SPECT(Q,x)$ are unique and do not depend on the choice of eigenbasis $B_Q(x)$.
\end{thm}

\revv{
The spectrum of quantum verification procedures plays an important role in the study of $\QMA$ and related complexity classes. It is closely related (but distinct) from the energy spectrum of local Hamiltonians. Several complexity classes with conditions on the spectrum have been studied, such as Polynomially Gaped QMA (PGQMA) in which the gap between the largest and second largest eigenvalue is inverse polynomial \cite{ABBS08}, and Exponentially Gaped QMA (EGQMA) in which the gap between the largest and second largest eigenvalue is  exponentially small \cite{DGF20}.}

\section{Functional Classes}
\label{Sec:FQMA}

 In this section we are interested in the set of states
 on which a quantum verification procedure $Q$ accepts with high probability. We may also be interested in the set of states on which $Q$  rejects with high probability. These sets will allow us to characterise  the functional analogs of the complexity classes introduced previously.
 
 There are two approaches to characterising these sets. The first is based on the notion of witness introduced in Definition \ref{Deff:QMA_and_coQMA}. The second uses the notion of  eigenbasis of a quantum verification procedure, see Theorems \ref{Thm:BlockStructure} and  \ref{Thm:UniquenessBlockStructure}.

 \subsection{Functional classes based on witnesses}
 
\begin{deff}
\label{Deff:Accepting_and_rejecting_states}
{\bf Accepting 
density matrices.} 
Let $Q=\{Q_n\}$ be a quantum verification procedure and fix $a\in[0,1]$.

We define the following relations
over binary strings and density matrices:
\begin{eqnarray}
R_{Q}^{\geq a}(x, \rho) = 1 &\text{ if }& \Pr [Q_{n}(x,\rho)=1]\geq a
\ ,\\
R_{Q}^{\leq a}(x, \rho) = 1 &\text{ if }& \Pr [Q_{n}(x,\rho)=1]\leq a
\ .
\label{EQ:RQgeqleq}
\end{eqnarray}

To simplify notation, we  denote
\begin{eqnarray}
R_{Q}^{\geq a}(x)&=&
\left\{ \rho  \ :\    R_{Q}^{\geq a}(x,\rho) =1\right\}
\ ,\\
R_{Q}^{\leq a}(x)&=&\left\{ \rho   \ :\    R_{Q}^{\leq b}(x, \rho) =1\right\}
\ .
\end{eqnarray}
\end{deff}

Consider an $(a,b)$-procedure $Q$. We are interested in the functional task of outputting a witness for $Q$, and in defining the corresponding complexity class.  This motivates the following definitions, where the subscript $W$ stands for ``Witness''.

\begin{deff}
\label{Deff:FQMA(a,b)_W}
{\bf Witness based definition of Functional $\QMA$ ($\FQMA_W$).}
 Let $a,b$ be functions as in Definition$~\ref{def:qvpab}$.
The {\em class} $\FQMA_{W} (a)$ is the set 
$\{R_{Q}^{\geq a}(x, \rho)\}$
where $Q$ is an $(a,b)$-procedure
\end{deff}

We can also define the functional analog of $\BQP$, for which we use Definition \ref{Def:BQP'}.

\begin{deff}{\bf Functional $\BQP$ ($\FBQP_W$).}
\label{Deff:FBQP}
The {\em class} $\FBQP_W(a)$ is the set of relations $\{R_{Q}^{\geq a}(x, \rho)\} \subseteq \FQMA_W(a)$
where $Q$ is an $(a,b)$--procedure, 
and where for each relation $R_{Q}^{\geq a}(x, \rho)$  there exists an efficiently preparable family of states $\{\rho(x)\}$ 
such that for all  $x$,
$\rho(x)\in R^{\geq a}_Q(x)$.
\end{deff}

We can also define the functional class total functional $\QMA$.

\begin{deff}{\bf Witness based definition of total Functional $\QMA$ ($\TFQMA_W$).}
\label{Deff:TFQMA-A}
Let $a: \N \rightarrow [0,1]$ be a polynomially time computable function.
The {\em class} $\TFQMA_W(a) \subseteq \FQMA_W(a)$ is the set of relations
$\{R_{Q}^{\geq a}(x, \rho)\}$
where $Q$ is an $a$-total procedure. 
\end{deff}

 \subsection{Functional classes based on eigenstates}
 
 The above definitions are natural. However they are not completely satisfactory for several reasons.
 
First,
we do not know whether $\FQMA_W (a)$ is independent of the threshold $a$. We obviously have that
$\FQMA_W (a') \subseteq \FQMA_W (a)$ if $a'\leq a$, but we do not know if the reverse is true, as $\FQMA_W (a)$ does not transform simply under amplification. 

Second, while the sets $R_{Q}^{\geq a}(x, \rho)$ are convex, they are not closed under linear combinations of pure states. That is, if the projectors onto $\vert \psi \rangle$ and $\vert \psi' \rangle$ belong to $R_{Q}^{\geq a}(x, \rho)$, then the projector onto the linear combination $a\vert \psi \rangle + b \vert \psi' \rangle$ does not necessarily belong to $R_{Q}^{\geq a}(x, \rho)$.

Third,  the very useful concept of eigenstate and eigenspace is not   captured by the functional classes based on witnesses. 
 This is illustrated in the following example:
\begin{expl}
\label{Example:No1}
Consider a function $\delta:\N \rightarrow [0,1/3[$ that decreases faster than $1/ {\poly}(n)$
for any polynomial $\poly(n)$, for instance $\delta(n)=2^{-n-2}$. The example consists of a $(1/3,2/3)$--quantum verification procedure $Q^{example}$ whose spectrum 
is the set
\begin{equation}
\SPECT(Q^{example},x) = \{ \frac{1}{3}, \frac{2}{3}-\delta^2(n), \frac{2}{3}+ \delta(n)\}
\end{equation}
where $n = |x|$.
\end{expl}

Let us denote by $\vert \psi_{2/3+\delta}\rangle$ and 
$\vert \psi_{2/3-\delta^2}\rangle$ two eigenstates with acceptance probability $2/3+\delta$ and $2/3-\delta^2$ respectively. Then the following state
\begin{equation}
\vert \psi  \rangle = \sqrt {\frac{\delta}{1+\delta}} \vert \psi_{2/3+\delta}\rangle
+ \sqrt {\frac{1}{1+\delta}} \vert \psi_{2/3- \delta^2}\rangle
\end{equation}
is a witness (it has acceptance probability $2/3$), hence belongs to $R_{Q}(x,)$. But it has 
 exponentially small overlap with the eigenspace $\HH_{Q}(x,\frac{2}{3}+ \delta(n))$. 
 
 These difficulties are avoided by the following definitions, based on the notion of eigenspace of a quantum verification procedure.

\begin{deff}
\label{Deff:Accepting_and_rejecting_subspaces}
{\bf Accepting and rejecting subspaces.} 
Let $Q=\{Q_n\}$ be a quantum verification procedure and fix $a,b\in[0,1]$.

We define the following relations over binary strings  and quantum states:
\begin{eqnarray}
\HH_{Q}^{\geq a}(x, \vert \psi \rangle) = 1 &\text{ if }& \vert \psi \rangle \in \SPAN (  \{ \HH_Q (x, p) \ :\  p\geq a\} ) \ ,
\nonumber\\
& & \\
\HH_{Q}^{\leq b}(x, \vert \psi \rangle) = 1 &\text{ if }& \vert \psi \rangle \in \SPAN (  \{ \HH_Q (x, p) \ :\   p\leq b\} )\ ,
\nonumber\\
\label{EQ:HHQgeqleq}
\end{eqnarray}
where  $\HH_Q(x,p)$ are the eigenspaces of $Q$ for $x$.

To simplify notation, we  denote
\begin{eqnarray}
\HH_{Q}^{\geq a}(x)&=&
\left\{ \vert \psi\rangle  \ :\    \HH_{Q}^{\geq a}(x, \vert \psi\rangle) =1\right\}\ ,\\
\HH_{Q}^{\leq b}(x)&=&\left\{ \vert \psi\rangle  \ :\    \HH_{Q}^{\leq b}(x, \vert \psi\rangle) =1\right\}\ ,
\end{eqnarray}
and we will generally express results in terms of the subspaces $\HH_{Q}^{\geq a}(x), \HH_{Q}^{\leq b}(x)$, rather than the corresponding relations.
\end{deff}

These relations provide the basis for an alternative definition of functional classes. These classes are denoted without the subscript $W$.

\begin{deff}
\label{Deff:FQMA(a,b)}
{\bf Functional QMA (FQMA).}
 Let $a,b$ be functions as in Definition$~\ref{def:qvpab}$.
The {\em class} $\FQMA(a,b)$ is the set 
$\{(\HH_Q^{\geq a}(x, \vert \psi \rangle), \HH_Q^{\leq b}(x, \vert \psi \rangle))\}$
of pairs of relations,
where $Q$ is an $(a,b)$-verification
procedure. 
\end{deff}

This definition is in terms of two relations ($\HH_Q^{\geq a}$ and 
$\HH_Q^{\leq b}$) for reasons which were presented in \cite{MS18}.

One can show that this definition of $\FQMA$ does not depend on the bounds $(a,b)$:

\begin{thm}
\label{Thm:FQMAAmp}
{\bf $\FQMA$ is independent of the soundness and completeness bounds \cite{MS18}.} 
\label{Thm:FQMASR}
Let $Q$ be a quantum verification procedure. 
Let
$a,b$ be a pair of functions as in Definition$~\ref{def:qvpab}$. For any $r\in \poly$ and 
 pair of functions $a',b'$ as in Definition$~\ref{def:qvpab}$
satisfying $a' < 1-2^{-r}$ and $b'> 2^{-r}$,  \rev{ we have
$\FQMA(a,b) \subseteq \FQMA(a',b')$.}
\end{thm}

Hence we take the traditional values $2/3$ and $1/3$:

\begin{deff}
\label{Deff:FQMA2/3}
We define the {\em class} $\FQMA$ as $\FQMA(2/3, 1/3)$.
\end{deff}

We   introduced $a$--total procedures in Definition
\ref{Deff:Totalqvp}. 
We can now define the corresponding functional classes.

\begin{deff}
\label{Deff:Total}
{\bf Totality.}
A pair of relations $(\HH^{\geq a}_Q(x, \vert \psi \rangle), \HH^{\leq b}_Q(x, \vert \psi \rangle))$ in $\FQMA(a,b)$ is called {\em total} if for all inputs $x$ there exists at least one witness
$\vert \psi \rangle $, i.e. if $\HH^{\geq a}_Q(x)$ is non empty.
\end{deff}

\begin{deff}
\label{Deff:TFQMA}
{\bf Total Functional  QMA (TFQMA).}
 Let $a,b$ be functions as in Definition$~\ref{def:qvpab}$.
The {\em class} $\TFQMA(a,b)$  is the set 
$\{(\HH_Q^{\geq a}(x, \vert \psi \rangle), \HH_Q^{\leq b}(x, \vert \psi \rangle))\}$
of pairs of total relations, i.e. the set of pairs of relations in $\FQMA$
where $Q$ is an $a$--total verification
procedure. 
\end{deff}

Theorem \ref{Thm:FQMAAmp} implies that we can take the thresholds to have their traditional values:

\begin{deff}
We define the class of total relations in $\FQMA$ as $\TFQMA=\TFQMA(2/3,1/3)$.
\end{deff}

A similar definition can be given for the functional analog of $\BQP$. However the definition of $\BQP$ is that there is a witness that can be efficiently prepared. Therefore functional $\BQP$ needs to be expressed in terms of the existence of an efficiently preparable witness. The definition $\FBQP_W$ thus seems the natural one in this case.

\subsection{Relation between the two definitions of Functional $\QMA$}

We discuss here how  the two definitions of Functional $\QMA$ are related.

\begin{thm}{\bf Relation between definitions of Functional $\QMA$\cite{MS18}.}\label{Thm:RelationsFQMA}
 Let $a,b$ be functions as in Definition$~\ref{def:qvpab}$ and let $Q$ be an $(a,b)$-procedure. Then 
\begin{enumerate}
\item 
we have the inclusion
\begin{equation}
 \HH_{Q}^{\geq a}(x) \subseteq  R_{Q}^{\geq a}(x)\ ,
\end{equation}
 where we view $\HH_{Q}^{\geq a}(x) $ not as a set of pure states, but as the 
 set of density matrices associated to these pure states. \rev{It also holds that the convex hull of $ \HH_{Q}^{\geq a}(x) $ is included in
$ R_{Q}^{\geq a}(x)$. In general these inclusions are strict, see Example \ref{Example:No1} for an illustration.}
 \item
In the other direction, if 
 $R_{Q}^{\geq a}(x)$ is non empty, then
$ \HH_{Q}^{\geq a}(x)$ is non empty.
\end{enumerate}
\end{thm}

\section{Reductions}
\label{Sec:Reductions}

\subsection{Reductions of procedures}

\begin{deff}{\bf Reduction.} 
\label{def:ReductionQ}

Let $a,b$ and $a',b'$ be pairs of functions  as in Definition$~\ref{def:qvpab}$.

Let $L$ be a language in $\QMA(a,b)$, and denote by 
 $Q=\{Q_{n} : n \in \N\}$ the associated $(a,b)$-quantum verification procedure.
 
Let $L'$ be a language in $\QMA(a',b')$, and denote by 
 $Q'=\{Q'_{n} : n \in \N\}$ the associated $(a',b')$-quantum verification procedure.
 
A reduction from $Q$ to $Q'$ 
is a pair $(f,\Phi)$
where $f:\N\to \N$ is a polynomial time computable function, and $\Phi$ is a family of efficiently implementable channel,
such that:
\begin{enumerate}
\item \label{Def:ReducQ_1}
For all $x\in L$, it holds that $f(x) \in L'$. 

In other words, for all $x$ such that there exists  $\vert \psi \rangle$ satisfying \begin{equation}
\Pr [Q_{\vert x\vert }(x,\vert \psi \rangle)=1]\geq a
\label{Eq:RedQa} 
\end{equation}
(i.e. $\vert \psi \rangle$  is a witness for $Q$ for $x$),
it holds that 
there exists $\vert \psi' \rangle$  satisfying
\begin{equation}
\Pr [Q'_{\vert f(x)\vert }(f(x),\vert \psi' \rangle)=1]\geq a'\ ,
\label{Eq:RedQb} 
\end{equation}
(i.e. $\vert \psi' \rangle$  is a witness
 for $Q'$ for $f(x)$).

\item \label{Def:ReducQ_2}
For all
 $x$, for all witnesses $\vert \psi' \rangle$ for $Q'$ for $f(x)$, it holds that 
$\Phi(x,\vert \psi '\rangle)$  is a witness for $Q$ for $x$.
 
In other words, for all $x$,  for all $\vert \psi' \rangle$ such that
\begin{equation}
\Pr [Q'_{\vert f(x)\vert }(f(x),\vert \psi' \rangle)=1]\geq a'
\ ,
\label{Eq:Red1} 
\end{equation}
it holds that 
\begin{equation} 
\Pr [Q_{\vert x\vert }(x,\Phi(x,\vert \psi' \rangle)=1]\geq a\ .
\label{Eq:Red2} 
\end{equation}
\end{enumerate}
\end{deff}

To illustrate this definition, we recall briefly Kitaev's construction that $k$-local-Hamiltonian is $\QMA$ complete\cite{KSV02}.

The $k$-local Hamiltonian problem, which is a promise problem, is defined as follows. The input $x$ is a classical description of a $k$-local Hamiltonian acting on $n$ qubits, where we recall that a $k$-local Hamiltonian  $H$ is the sum of polynomially many  Hermitian matrices (whose norm is bounded by a polynomial) that act on only $k$ qubits. The input also contains two numbers $a<b$, such that $b-a \in 1/\poly$. The problem is to determine whether the smallest eigenvalue of this Hamiltonian is less than $a$ or greater than $b$, promised that one of these is the case. 

Kitaev showed that if $Q\in \QMA$, then for every instance $x$, there exists an instance $f(x)$ of $k-$local Hamiltonian --which we write $H(f(x))$-- such that:\\
1) if $x\in L$ there exists a low energy eigenstate of  $H(f(x))$;\\
2) if $x\notin L$ then $H(f(x))$ does not have low energy eigenstates.

Recall that the witnesses in Kitaev's construction have the form
\begin{equation} 
\frac{1}{\sqrt{L+1}}\sum_{j=0}^L U_j (\vert \psi \rangle \vert 0^k\rangle )\vert j\rangle
\end{equation}
where the first register is the witness space, the second register is the ancilla space, and the third register is the time counter, $L$ is the number of gates in $Q$, \rev{ and $U_j$ is the unitary that implements the $j$ first  gates of the procedure $Q$ (with $U_0$ the identity).}

A reduction from $Q$ to $k-$local Hamiltonian consists of:
1) the function $f(x)$; 2) the channel $\Phi$ which consists in 
measuring the time counter of the witness to obtain $j$ and then undo the computation conditional on $j$ (i.e. apply $U_j^\dagger$).

\begin{thm}
\label{ThmQMAReduc}
{\bf } 

Let $a,b$, $a',b'$, $L$, $L'$, $Q$, $Q'$ be functions, languages and procedures as in Definition$~\ref{def:ReductionQ}$.

Suppose $(f,\Phi)$ is a reduction from $Q$ to $Q'$.

Then there exists an $(a,b)$--procedure $Q^R$ such that 
\begin{itemize}
\item The language associated to $Q^R$ is $L$;
\item If $\vert \psi'\rangle$ is a witness for $Q'$ for $f(x)$, then 
$\vert \psi'\rangle$ is a witness for $Q^R$ for $x$, or expressed differently
\begin{equation}
R^{\geq a'}_{Q'}(f(x))\subseteq R^{\geq {a}}_{Q^R}(x)
\label{RelationReduction}
\end{equation}
\end{itemize}
\end{thm}

\begin{proof} Trivial.
The procedure $Q^R$ is given by 
$Q^R(x,\vert \psi\rangle)= Q(x,\Phi(\vert \psi\rangle))$.
\end{proof}

We note that it would be interesting to generalise slightly the definition of reduction of procedures, and only require that the channel $\Phi$ succeed with probability greater than $1/\poly$. For instance in the case of a reduction  from $Q$ to $k-$local Hamiltonian  the channel $\Phi$ could consists in 
measuring the time counter of the witness and rejecting except if the result is $j=0$. However as the above discussion of $k-$local Hamiltonian  suggests, a chanel $\Phi$ that always succeeds seems sufficient. In the present work we will suppose that $\Phi$  always succeeds. The case where $\Phi$  is probabilistic would be interesting to study, but is a non--trivial generalisation. For instance it is not obvious whether Theorem \ref{ThmQMAReduc} below holds in this case.

\begin{deff}{\bf QMA Completeness.} 
\label{def:Completness}

Let $a,b$  be functions as in Definition$~\ref{def:qvpab}$.

Let $L$ be a language in $\QMA(a,b)$, and denote by 
 $Q=\{Q_{n} : n \in \N\}$ the associated $(a,b)$-quantum verification procedure.  
 
The procedure $Q$ is $\QMA$-complete if, for every $a',b'$ functions as in Definition$~\ref{def:qvpab}$, for every language $L'$  in $\QMA(a',b')$ with
 $Q'=\{Q'_{n} : n \in \N\}$ the associated $(a',b')$-quantum verification procedure, there exists a reduction from $Q'$ to $Q$.
\end{deff}

\subsection{Robust reductions}
\label{subsec:robreduc}

For  many questions concerning $\QMA$ the above notion of
reduction is sufficient. However Example \ref{Example:No1} shows that it may not be appropriate in some cases.
Suppose there exists a reduction $(f,\Phi)$ from a procedure $Q^{example}$  to procedure $Q$. This reduction tells us how the witness of $Q^{example}$ (the eigenstates which accept with probability $2/3+\delta$) transform under $\Phi$. But it tells us nothing about how eigenstates with acceptance probability $2/3 - \delta^2$ transform under $\Phi$. 
But these eigenstates are operationally \rev{ indistinguishable (except possibly by using the structure of $Q^{example}$) from the genuine witnesses, since their acceptance probability is so close to $2/3$.}

In Section \ref{Sec:RobReducQMATFQMA} we will need a slightly stronger form of reduction which we call robust reduction, and which does not suffer from this problem.

\begin{deff}{\bf Robust Reduction.} 
\label{def:Eps-Reduction}

Let $a,b$, $a',b'$, $L$, $L'$, $Q$, $Q'$ be functions, languages and procedures as in Definition$~\ref{def:ReductionQ}$.

Recall that $a'-b'\geq 1/q'$ for some $q'\in \poly$. Let $\epsilon\in 1/\poly$ be such that $\epsilon \leq 1/(2q')$. (As a consequence $a'-\epsilon$ lies between $a'$ and $b'$, at a distance $\geq 1/\poly$ from either bound).

A robust reduction from $Q$ to $Q'$ with parameter $\epsilon$ is a reduction from $Q$ to $Q'$ via the pair $(f,\Phi)$ where condition \ref{Def:ReducQ_2} of  Definition$~\ref{def:ReductionQ}$ is replaced by the stronger condition

\begin{itemize} 
  \item[2'.]
For all
 $x$, for all $\vert \psi' \rangle$ such that
\begin{equation}
\Pr [Q'_{\vert f(x)\vert }(f(x),\vert \psi' \rangle)=1]\geq a'-\epsilon
\ ,
\label{Eq:Red1-1} 
\end{equation}
it holds that 
\begin{equation} 
\Pr [Q_{\vert x\vert }(x,\Phi(x,\vert \psi' \rangle)=1]\geq a\ .
\label{Eq:Red2-1} 
\end{equation}  
\end{itemize}
\end{deff}

Note that a robust reduction with parameter $\epsilon=0$ is a reduction according to Definition$~\ref{def:ReductionQ}$.

\begin{prop}\label{Prop:Reductions}
The following hold for all  $\epsilon, \epsilon' \geq 0$:
\begin{enumerate}
\item 
A robust reduction with parameter $\epsilon$ is also a robust reduction with parameter $\epsilon'$ for all $ 0\leq \epsilon'\leq \epsilon$.
\item Transitivity of reductions.
If  there exists a robust reduction from  $Q$ to $Q'$ with parameter $\epsilon$,
and if there exists a robust reduction from  $Q'$ to $Q''$ with parameter $\epsilon'$,
then there exists a robust reduction from  $Q$ to $Q''$ with parameter $\epsilon'$.
\end{enumerate}
\end{prop}

\begin{proof} Immediate.
\end{proof}

We note that $k$--local Hamiltonian is also $\QMA$ complete using the stronger notion of robust reduction. This follows from the promise that that the smallest eigenvalue of the $k$--local Hamiltonian is less than $a$ or greater than $b$. \rev{This, together with the fact that we need robust reductions to prove Theorem \ref{ThmQMA=QMAcoQMA}, suggests that robust reductions may in fact be the natural notion to use when considering reductions of procedures.}

\section{Eigenspace preserving maps.}
\label{Sec:E-maps}

Reductions (robust or not) do not preserve the eigenspace structure of a quantum verification procedure. 
The following definitions and results repeated from \cite{MS18}, define a mapping between procedures that preserves this structure.

\begin{deff}{\bf Eigenspace preserving map of quantum verification procedures.}
\label{Deff:StrongReduction}
Let $Q$ and $Q'$ be two quantum verification procedures with the same witness space dimensions. We say that  there exists \emph{an eigenspace preserving map from $Q$   to $Q'$} if there exists a polynomial time
 computable strictly  increasing 
  family of functions $\{f_n:n \in \N\}$, $f_n:[0,1]\to [0,1]$, 
such that
for all $x\in \{0,1\}^*$ with $n=\vert x \vert$:  
\begin{enumerate}
\item
there exists a basis $B_Q(x)=\{\vert \psi_i\rangle \}$ of the witness Hilbert space $\HH_m$
which is a joint eigenbasis of $Q$ and $Q'$ for $x$;
\item 
for all $\vert \psi_i\rangle \in B_Q(x)$ with
$p_i=\Pr [Q_{n}(x,\vert \psi_i \rangle)
=1]$ the acceptance probability of $\vert \psi_i\rangle$ for $Q_n$, 
and $p'_i=\Pr [Q'_{n}(x,\vert \psi_i \rangle)
=1]$ the acceptance probability of $\vert \psi_i\rangle$ for $Q'_n$,  
it holds that $p'_i=f_n(p_i)$.
\end{enumerate}
In what follows we will refer to an eigenspace preserving map simply as an \emph{e--map}.
\end{deff}

The reason why we require that the functions $f_n$ 
be polynomial time computable 
is because we wish that the soundness and completeness thresholds  of $Q$ be mapped onto the soundness and completeness thresholds of $Q'$, where we recall
that the soundness and completeness thresholds must be polynomial time computable, see Definition \ref{def:qvpab}.
That is, if $Q$ is an $(a,b)$-procedure that e--maps to $Q'$
via $f_n$, 
then $Q'$ is an $a',b'$-procedure with $a'(n)=f_n(a(n))$ and $b'(n)= f_n(b(n))$.

\begin{prop}{\bf \cite{MS18}}
\label{Prop:eMaps}
The following hold:
\begin{enumerate}
\item Conservation of eigenspaces. If $Q$ e-maps to $Q'$ via $\{f_n\}$, then the eigenspace of $Q$ for $x$ with acceptance probability $p$ equals the eigenspace of $Q'$ for $x$ with acceptance probability $f_n(p)$:
\begin{equation}
\HH_Q(x,p) = \HH_{Q'}(x,f_{\vert x \vert}(p))\ .
\end{equation}
\item Conservation of accepting and rejecting subspaces. 
 Let $Q$ be an $(a,b)$--procedure and let 
 $Q'$ be an $(a',b')$, and suppose that $Q$
  e-maps to $Q'$ via $\{f_n\}$ with $a'(n)=f_n(a(n))$ and
  $b'(n)=f_n(b(n))$, then
\begin{eqnarray}
\HH_Q^{\geq a}(x) &=& \HH_{Q'}^{\geq a'}(x)\nonumber\\
\HH_Q^{\leq b}(x) &=& \HH_{Q'}^{\leq b'}(x)
\end{eqnarray}
\item Transitivity of reductions.
If $Q$ e--maps to $Q'$,
and $Q'$ e--maps to $Q''$,
then $Q$ e--maps to $Q''$.
\end{enumerate}
\end{prop}


Eigenspace preserving maps are characterised by the family of functions $\{f_n\}$. What freedom do we have in choosing the functions $f_n$? The next section introduced iterative procedures, which are of interest in themselves. Iterative procedures are then used in section \ref{Sec:GenAmp-2} to show that
\revv{we have a lot of freedom in choosing the functions $f_n$.}

\section{ Iterative Procedures}\label{Sec:GenAmp}

The amplification result for $\FQMA$, Theorem \ref{Thm:FQMASR}, is based on the method introduced in \cite{MW05} in which a quantum verification procedure is run repeatedly backwards and forwards. Here we generalise this approach.


\begin{deff}{\bf Iterative Procedures.}
\label{Deff:GAmpProc}
Let $Q=\{Q_{n} : n \in \N\}$ be a quantum verification procedure with parameters $m$ (the size of the witness space) and $k$ (the size of the ancilla space). 

Recall that the Hilbert space $\HH$ on which $Q_n$ acts can be decomposed into $\HH=\HH^w_m\otimes \HH_k^a$ 
where $\HH^w_m$ is the witness Hilbert space comprising $m$ qubits, and $\HH_k^a$ is the ancilla Hilbert space comprising $k$ qubits.  The Hilbert space can also be decomposed into
$\HH=\HH_1^{out}\otimes \HH^{rest}_{m+k-1}$ 
where $\HH^{out}_1$ of dimension 2 
is the qubit which is measured at the end of the quantum verification procedure, and $ \HH^{rest}_{m+k-1}$ is the Hilbert space of the remaining $m+k-1$ qubits.

Define the following two projectors
\begin{eqnarray}
\Pi_0&=&I_m\otimes \vert 0^k\rangle \langle 0^k\vert\\
\Pi_1&=& V_x^\dagger \left(  \vert 0\rangle \langle 0\vert \otimes I_{m+k-1}^{rest}\right)V_x
\label{Eq:Pi1}
\end{eqnarray}
where $V_x$ is the unitary transformation realised by the quantum verification procedure.
That is $\Pi_0$ projects onto valid inputs to the quantum verification procedure, and $\Pi_1$ projects onto states on which the quantum verification procedure will accept with certainty.

Let $N\in \poly$ and let 
\begin{eqnarray}
G &=& \{g_n : \{0,...,N(n)\}\rightarrow [0,1] \ :\ 
x\in  \N \} \ \ 
\end{eqnarray}
 be a family of polynomially time computable functions.  
 
 The {\em $(N,G)$--iterative procedure obtained from $Q$} is a quantum verification procedure $Q_{it}$ 
with the same witness space dimension $m$ as $Q$, 
and an ancilla space which we decompose into a first space of $k$ qubits (the same dimension as for $Q$) and an additional working space of dimension $k'$. This second space can be classical. The input space of $Q_{it}$ is thus $\HH_m\otimes \HH_k\otimes \HH_{k'}$. We view the projectors $\Pi_0$ and $\Pi_1$ as acting on the space $\HH^w_m\otimes \HH_k$.

On input $(x,\vert \psi \rangle\vert 0^k \rangle\vert 0^{k'} \rangle)$ 
with $\vert \psi \rangle \in \HH_m$, $Q_{it}$ acts as follows:
\begin{enumerate}
\item
Repeat until $N+1$ measurement outcomes have been registered:
\begin{enumerate}
\item measure $\{ \Pi_0, 1 - \Pi_0\}$; \label{list2(b)}
\item measure $\{ \Pi_1, 1 - \Pi_1\}$.
\end{enumerate}
(Note: the result of the first measurement will necessarily be $1$ since the initial state of the first $m+k$ qubits ($\vert \psi \rangle \vert 0^k\rangle$) is an eigenstate of $\Pi_0$ with eigenvalue $1$).
\item
Inspect the sequence of results. 
\begin{enumerate}
\item If results $i$ and $i+1$ differ set $z_i=0$; 
\item while if results $i$ and $i+1$ are equal set $z_i=1$.
\end{enumerate}
\item Set $s=\sum_{i=1}^N z_i$.
\item Generate a random bit $c\in \{0,1\}$ satisfying $\Pr [c=1]=g_{\vert x\vert} (s)$ and 
$\Pr [c=0]=1- g_{\vert x\vert}(s)$. 
\item Output $c$.
\end{enumerate}
\end{deff}

Note that, in view of step 2,  the functions $g_n(s)\in[0,1]$ can be interpreted as probability distributions over the set $\{0,1\}$. One can extend the above generalised amplification procedure and take the $g_n(s)$ to be probability distributions over the alphabet $\{0,...,d-1\}$, in which case the generalised amplification would yield a $d$-outcome quantum verification procedure. 

A particularly interesting case of iterative procedure is when the $g_n$ are threshold functions:
\begin{eqnarray}
g_n(s)&=&0 \mbox{ if } s<s_0(n)\nonumber\\
&=&1  \mbox{ if } s\geq s_0(n)
\end{eqnarray}
for some polynomial time  computable function $s_0:\N \to \N$ with
$s_0(n)\in \{0,...N(n)\}$. The $\QMA$ amplification described in \cite{MW05} is based on a threshold iterative procedure.

In the following we will denote by $f(k; N,p)$ the probability mass function of the binomial 
distribution with parameters $N$ and $p$. \rev{ Thus if 
\begin{equation}
X\sim B(N,p)\ ,
\label{Eq:Binomial}
\end{equation}
} 
then 
\begin{equation}
f(k; N,p) = \Pr [X = k] = \binom{N}{k} p^{k} (1-p)^{N-k} \ .
\end{equation}
In addition, for any $g:\{0,...,N\}\to [0,1]$, we denote by $P_g$ the function
\begin{eqnarray}
P_g&:&[0,1]\rightarrow [0,1] \nonumber\\
& :& p \rightarrow P_g(p) =  E[g(X)]=\sum_{k=0}^N f(k; N,p) g(k) \ .
\nonumber\\
\label{Eq:Pg}
\end{eqnarray}

\begin{thm}{\bf Properties of Iterative Procedures.}\label{Thm:PropertiesGenAmpl}
Let $Q=\{Q_{n} : n \in \N\}$ be a quantum verification procedure, let $(N,G)$ be as in Definition \ref{Deff:GAmpProc}, and let 
 $Q_{it}$ be the $(N,G)$-iterative procedure obtained from $Q$. Then the following hold:
 \begin{enumerate}
 \item There exists a basis $B_Q(x) = \{\vert \psi_i\rangle\}$  of the witness Hilbert space $\HH_m$ which is a joint eigenbasis of $Q$ and $Q_{it}$ for $x$.
\item 
Let $\vert \psi_i \rangle$ be an eigenstate of $Q$ for $x$ (and therefore also an eigenstate of $Q_{it}$ for $x$). Denote by $p_i$ and $p'_i$ the acceptance probabilities of $\vert \psi_i \rangle$ for $Q$ and $Q_{it}$ respectively.

If the initial state of the witness space is $\vert \psi_i \rangle$, then the following hold:
\begin{enumerate}
\item 
The variables $z_i$ defined in step 2 are $0-1$ i.i.d. random variables satisfying
$\Pr [z_i=1]=p_i$ (i.e. equals the acceptance probability of $\vert \psi_i \rangle$).
 \item The acceptance probabilities of the iterative procedure is given by
 \rev{
\begin{eqnarray}
 p'_{i} &=& P_{g_{n}}(p_i) \nonumber\\
 &=&\sum_{k=0}^N f(k; N,p_i) g_{n}(k) 
\label{eq:P'i}
\end{eqnarray}
where $n=\vert x \vert$.}
\item  Whenever at step \ref{list2(b)} of Definition \ref{Deff:GAmpProc}  the outcome of measurement $\{ \Pi_0, 1 - \Pi_0\}$ is $1$, the state of the witness space is the original witness $\vert \psi_i \rangle$.\label{QNDProperty2}
\end{enumerate}
\end{enumerate}
\end{thm}

\begin{proof} Not given. Simple application of the construction in \cite{MW05}. \end{proof}

\section{ \rev{Iterative Procedures and Eigenspace Preserving Maps}}\label{Sec:GenAmp-2}

\rev{
We now show how iterative procedures can be used to construct e-maps. 
}

\rev{
The results in the present section are purely classical, and can be phrased as the following problem. Given $N$ draws from a coin with  bias $p$, generate a new coin with bias $f(p)$, where $p$ is unknown, but the function $f:[0,1]\to [0,1]$ is known. To this end introduce a function $g:\{0,...,N\}\to [0,1]$, and upon finding that $k$ of the coins come out head, toss a coin with bias $g(k)$. The coin so obtained will have bias 
$p'=\sum_{k=0}^N f(k; N,p) g(k)$. This is Eq. \eqref{eq:P'i}.
}

In addition we need to impose the further  conditions: the function $f$ is strictly increasing (see Definition \ref{Deff:StrongReduction} of e-maps), and both $f$ and $g$ are polynomial time computable (see Definitions \ref{Deff:StrongReduction}  and \ref{Deff:GAmpProc}).

First we show in Theorem \ref{Thm:GAmpStrongReduc} that, under the natural condition that the function $g$ is increasing and non constant, the function $f$ is strictly increasing.

\revv{
Second we show in Theorem \ref{Thm:SuperStrongAmpl} that by choosing the function $g$ appropriately, we have a lot of freedom in choosing the function $f$. More precisely we show that for any $M$, if $N$ is large enough we can choose $g$ so that $f$ is strictly increasing and passes through $M$ predefined points $(s_i,t_i)$, i.e. $f(s_i)=t_i$, where $i=1,...,,M$ (to which we add the points $(0,0)$ and $(1,1)$). The main remaining condition is that the $s_i$ have inverse polynomial gaps.
}

\begin{thm}{\bf Iterative procedures with  increasing functions $g_n$ are eigenspace preserving maps.}\label{Thm:GAmpStrongReduc}
Given a quantum verification procedure $Q$, and $(N,G=\{g_n\})$ as in Definition \ref{Deff:GAmpProc} with the functions $g_n:\{0,...,N\}\rightarrow [0,1] $   increasing and non constant (by this we mean that $g(k+1)\geq g(k)$ and $g(N)>g(0)$), then $Q$ 
e--maps to the $(N,G)$-generalised amplification procedure $Q_{it}$ obtained from $Q$.
\end{thm}

\begin{proof}
The condition that $g$ is increasing and non constant  implies that there is some $k_0\in \{0,...,N-1\}$ such that
$g(k_0+1) > g(k_0)$ (with a strict inequality).

We need to show that under the conditions that the functions $g_n$ are polynomial time computable, increasing and non constant, $p'_i$ given in Eq. \eqref{eq:P'i}, viewed as a function of $p_i$, is a polynomial time computable strictly increasing function.

It is therefore sufficient to show that for any  polynomial time  computable strictly increasing function $g$,  the function $P_g(p)$ defined in Eq. \eqref{Eq:Pg}
is  a polynomial time  computable strictly increasing function of $p$. 

The fact that $P_g(p)$ is  polynomial time  computable follows immediately from the fact that the functions $g_n$ are polynomial time computable. 

We now show that $P_g(p)$ is strictly increasing.

First note that for all $p\in[0,1]$, we have
\begin{equation}
1= \sum_{k=0}^N f(k; N,p)\ .
\label{Eq:Bin=1}
\end{equation}
Taking the derivative of Eq. \eqref{Eq:Bin=1} with respect to $p$, we have
\begin{eqnarray}
0&=& \sum_{k=0}^N \binom{N}{k} p^{k-1} (1-p)^{N-k-1} (k-Np)\ .\ 
\label{Eq:DerBin0}
\end{eqnarray}
(Note that as written  this expression is well defined only  for $0 < p < 1$. However the limits $p\to 0$ and $p\to 1$ are finite. We extend Eq. \eqref{Eq:DerBin0} and all expressions below to $p=0$ and $p=1$ by replacing the ill--defined terms by their limit).

Hence 
\begin{eqnarray}
&\sum_{k=0}^{\lfloor Np \rfloor} \binom{N}{k} p^{k-1} (1-p)^{N-k-1} (Np-k)& \nonumber\\
&= 
\sum_{k=\lceil Np \rceil}^N \binom{N}{k} p^{k-1} (1-p)^{N-k-1} (k-Np)&\ ,\quad 
\label{Eq:pos-neg}
\end{eqnarray}
where we note that all terms under summation signs are positive.

Now take the derivative of Eq. \eqref{Eq:Pg} with respect to $p$ to obtain
\begin{eqnarray}
&\frac{\partial P_g}{\partial p} &=
 \sum_{k=0}^N \binom{N}{k} p^{k-1} (1-p)^{N-k-1} (k-Np)g(k)\nonumber\\
  &=&
  -\sum_{k=0}^{\lfloor Np \rfloor} \binom{N}{k} p^{k-1} (1-p)^{N-k-1} (Np-k) g(k)
  \nonumber\\
  & & +
 \sum_{k=\lceil Np \rceil}^N \binom{N}{k} p^{k-1} (1-p)^{N-k-1} (k-Np) g(k)\nonumber\\
   & & \label{Eq:StrictIncrease1}\\
 &\geq &
  -\sum_{k=0}^{\lfloor Np \rfloor} \binom{N}{k} p^{k-1} (1-p)^{N-k-1} (Np-k) g(\lfloor Np \rfloor)
  \nonumber\\
 & &  +
 \sum_{k=\lceil Np \rceil}^N \binom{N}{k} p^{k-1} (1-p)^{N-k-1} (k-Np) g(\lceil Np \rceil)\nonumber\\
   & & \label{Eq:StrictIncrease2}\\
&=&
 \sum_{k=\lceil Np \rceil}^N \binom{N}{k} p^{k-1} (1-p)^{N-k-1} (k-Np) 
 \nonumber
 \\
 & & \quad \quad \quad \quad \times
  \left( g(\lceil Np \rceil) - g(\lfloor Np \rfloor) \right)
  \nonumber
 \\
   & & \label{Eq:StrictIncrease3}\\
  & & \geq 0 \label{Eq:StrictIncrease4}
\end{eqnarray}
where for the last equality we have used Eq. \eqref{Eq:pos-neg}.

Now note that either $k_0\in\{0,...,\lfloor Np \rfloor-1\}$, or $k_0=\lfloor Np \rfloor$, or 
$k_0\in\{\lceil Np \rceil,...,N-1\}$. In the second case, there is a strict  inequality when going from Eq. \eqref{Eq:StrictIncrease3} to Eq. \eqref{Eq:StrictIncrease4}. In the first case or the third case, there is a strict  inequality when going from Eq. \eqref{Eq:StrictIncrease1} to Eq. \eqref{Eq:StrictIncrease2}. Therefore under the condition that $g$ is increasing and nonconstant, $\frac{\partial P_g}{\partial p} >0$, i.e.
 $P_g$ is strictly increasing.
\end{proof}

This result shows that there are many e--maps. We now show that  
\revv{we have a lot of freedom in choosing}
the functions $\{f_n\}$ that define the e--map.

\begin{thm}
\label{Thm:SuperStrongAmpl}
Let $Q=\{Q_n\}$ be a quantum verification procedures with witness Hilbert space dimension $m(n)$. Let $u, v, M \in \poly$, 
and denote $\epsilon(n) = 1/u(n)$ and $\delta(n)=\exp(-v(n))$. 
Let $S$ and $T$ denote the following sets:
\begin{eqnarray}
S&=&\{S_n :n\in \N\}\nonumber\\
S_n &=& (s_0,s_1,...,s_M,s_{M+1})\nonumber\\
s_i &\in & [0,1]
\ , \ s_0=0\ , \ s_{M+1}=1\nonumber\\
\forall i \ &:&\ s_i - s_{i-1} \geq \epsilon >0 \label{Eq:pbound}\\
\forall i \ &:&\  s_i \mbox{ is polynomial time  computable}\nonumber
\end{eqnarray}
and
\begin{eqnarray}
T&=&\{T_n :n\in \N\}\nonumber\\
T_n &=& (t_0,t_1,...,t_M,t_{M+1})\nonumber\\
t_i &\in & [0,1]
\ , \ t_0=0\ , \ t_{M+1}=1\nonumber\\
\forall i \ &:&\ t_i - t_{i-1} \geq \delta >0 \label{Eq:qbound}\\
\forall i \ &:&\  t_i \mbox{ is polynomial time  computable}\nonumber
\end{eqnarray}
where $M=M(n)$.
(Note that the definitions of $S$ and $T$ differ only by Eqs. 
 \eqref{Eq:pbound} and \eqref{Eq:qbound} 
 which ensure that $s_i$ and $t_i$ are ranked in increasing order, but with different gaps).

Then there exists a quantum verification procedure $Q'$ such that $Q$ e-maps to $Q'$ and the polynomial time  computable strictly increasing functions $f_n$ that define the e-map satisfy $f_n(s_i)=t_i$ for all $i\in \{0,...,M+1\}$.
\end{thm}

The fact that the $s_i$'s have inverse polynomial gaps, see Eq. \eqref{Eq:pbound}, is the main remaining constraint on the choice of the functions $f_n$. When we use Theorem 
\ref{Thm:SuperStrongAmpl} this will however not be a limitation, as typically the $s_i$'s  will be taken to be soundness and/or completeness bounds, which always have inverse polynomial gaps.
\revv{
We do not expect that this condition can be lifted because the spectra of quantum verification procedures are related to their complexity. In particular it seems unlikely that one could expand exponentially small intervals in the spectrum into inverse polynomial intervals. Indeed it is known that when the gap between the largest and second largest eigenvalue of the quantum verificaton procedure  and/or 
the gap between the completeness and soundness bounds are taken to be exponentially small, one has new, much larger, complexity classes such as PP and PSPACE, see 
 \cite{DGF20} for details and exact definitions. If we could take the $s_i$'s to have exponentially small gaps one could probably show that these complexity classes collapse with QMA.
 }

\begin{proof}

{\bf Outline of the proof.}
We will use Theorem \ref{Thm:GAmpStrongReduc}.
We will show that there exists $(N,G=\{\bar g_n\})$ as in Definition \ref{Deff:GAmpProc}, with the functions $\bar g_n:\{0,...,N\}\rightarrow [0,1] $  polynomial time  computable and strictly increasing, such that the  $(N,G)$-generalised amplification procedure $Q'$ obtained from $Q$ satisfies the conditions of Theorem \ref{Thm:SuperStrongAmpl}. The  polynomial time  computable strictly increasing functions $f_n$ that define the e-map will be given by 
$P_{\bar g_n}$.

\rev{
It is convenient to extend the discrete functions $\bar g_n$ to continuous functions $g_n:[0,N]\to[0,1]$ such that $\bar g_n$ coincides with $g_n$ on the integers. The functions 
$g_n$ are taken to be continuous, piece wise linear and to satisfy $g_n(N s_i)=t_i +\lambda_i$ where $\lambda_i$ are small parameters that we fix later in the proof. We then show that the conditions $f_n(s_i)=t_i$ yield a set of $M$ linear equations for $\lambda_i$. We show that for large $N$ the solutions $\lambda_i$ of these equations are exponentially small. The smallness of the $\lambda_i$ is also a sufficient condition for the functions $g_n$ to be strictly increasing, which implies that the $f_n$ are strictly increasing (see Theorem \ref{Thm:GAmpStrongReduc}). 
}

{\bf Definition of $\bar g_n$.}

Choose a function $N\in \poly$. We will determine how large $N$ must be at the end of the proof.

We denote by
\begin{eqnarray}
\Delta_i &=&[N(s_i-\frac{\epsilon}{3}), N(s_i+\frac{\epsilon}{3})]\ , \nonumber\\
\Delta_0 &=& [0, N\frac{\epsilon}{3}]\ ,
\nonumber\\
\Delta_{M+1} &=&  [1 - N\frac{\epsilon}{3},1]\ ,\nonumber\\
\Xi_i &=& ( N(s_{i}+\epsilon/3) , N(s_{i+1}-\epsilon/3) )\ ,
\end{eqnarray}
 where $ i=1,...,M$.
And we denote by
\begin{eqnarray}
\bar{\Delta}_i &=& \{ n\in \N : n \in \Delta_i\}\ ,
\nonumber\\
\bar \Xi_i &= &\{ n\in \N \ :\ n \in \Xi_i\}
\end{eqnarray}
the integers belonging to these intervals.

Let $X_i\sim B(N,s_i)$ \rev{be distributed as a binomial}. For all $s_i$, $i=1,...,M$,  denote
\begin{equation}
\Pr [  X_i \in \bar \Delta_i ] = 1 - \mu_i\ .
\end{equation}
Denote
\begin{equation}
\mu = \max_i \{\mu_i\} \ .
\end{equation}
Note that by taking $N$ sufficiently large, we can ensure that $\mu$ is exponentially small, and in particular smaller than $\delta$.

Let $(\lambda_0, \lambda_1,...\lambda_M,\lambda_{M+1})$ be parameters we will fix later, except that we fix $\lambda_0=0$ and $\lambda_{M+1}=0$.
Denote
\begin{equation}
\Lambda = \max_i \{\vert \lambda_i\vert \} \ .
\end{equation}

Let 
\begin{equation} \sigma = \frac{\delta}{2 \epsilon}\ .
\end{equation}

Consider the continuous, piece wise linear functions $g_n:[0,N] \to \R$ which on $ \Delta_i$ are given by
\begin{eqnarray}
g_n(z) &=& (t_i + \lambda_i) + \sigma (t-N s_i) \ ,  \ z\in \Delta_i \ ,
\label{Eq:gx1}
\end{eqnarray}
and which are linear 
outside of the intervals $\Delta_i$. Continuity then implies that  for $z \in \Xi_i$, $i=0,...,M$, one has the expression:
\begin{eqnarray}
g_n(z) &=& \frac{ (t_{i+1}+ \lambda_{i+1}) -   (t_i + \lambda_i)  - 2 \epsilon \sigma/3}
{N(s_{i+1} - s_i - 2 \epsilon/3)} \times \nonumber\\
& & \times (z-N (s_i + \epsilon/3))  
\nonumber\\ & & +(t_i +\lambda_i + \sigma\frac{\epsilon}{3})\ .
\label{Eq:gx2}
\end{eqnarray}

We denote by $\bar g_n:\{0,...,N\}\to [0,1]$  the functions which coincide with the functions $g_n$  on the integers.

{\bf Strategy for proving the existence of an e--map with the desired properties.}

We wish to show that for  $N$ large enough, one can choose the parameters $\lambda_i$ such that:  
\begin{enumerate}
\item $Q$ e-maps to the $(N,G=\{g_n\})$ generalised amplification procedure obtained from $Q$;
\item and that for all $i$
\begin{equation}
 t_{i} = \sum_{k=0}^N f(k; N,s_i) \bar g_n(k) \ .\label{Eq:qfg}
\end{equation}
\end{enumerate}

Note that if $\lambda_i=0$, then $\bar g_n (N s_i)=t_i$.
Hence if the probability mass functions $f(k; N,s_i)$ 
are strongly peaked around $N s_i$ (which is the case if $N$ is large enough), one expects that Eqs. \eqref{Eq:qfg} can be satisfied for small $\lambda_i$. We show that this is indeed the case below.

{\bf Boundary values.}

The functions $g_n$ satisfy $g_n(0)=0$ and $g_n(N)=1$.

{\bf Condition for strict increase.}

Consider the slope of the functions $g_n$. Over the intervals $\Delta_i$ the slope is $\sigma$ which is strictly positive.

Between the intervals $\Delta_i$, the slope is
equal to 
\begin{equation}
\frac{ (t_{i+1}+ \lambda_{i+1}) -   (t_i + \lambda_i)  - 2 \epsilon \sigma/3}
{N(s_{i+1} - s_i - 2 \epsilon/3)} \ ,
\end{equation}
see Eq. \eqref{Eq:gx2}. This is strictly positive if
$(t_{i+1}+ \lambda_{i+1}) -   (t_i + \lambda_i)  - 2 \epsilon \sigma/3$ is strictly positive.
Now note that 
\begin{eqnarray}
(t_{i+1}+ \lambda_{i+1}) -   (t_i + \lambda_i)  - 2 \epsilon \sigma/3 
&>& 2  \delta/3 - 2 \Lambda\nonumber\\
\end{eqnarray}
 
Hence a sufficient condition for the functions $g_n$ being strictly increasing is
\begin{equation}
\Lambda < \delta/3\ .
\label{Eq:boundlambda}
\end{equation}

{\bf Computing $\lambda_i$.}

 Eqs. \eqref{Eq:qfg} are a set of $M$ linear equations in the $\lambda_i$, with coefficients that can be efficiently computed. Hence,  since $s_i$ and $t_i$ are polynomial time computable,  the functions $g_n$ are polynomial time  computable.

It remains to show that for $N$ large enough solving Eqs. \eqref{Eq:qfg}  for the $\lambda_i$ yields solutions that satisfy Eq. \eqref{Eq:boundlambda}.

{\bf Properties of invertible matrices.}

Let $A\in \R^{M\times M}$ be an $M \times M$ matrix, and denote by $A_{ij}$ its elements. 
Denote by
\begin{equation}
\| A\|_\rev{\infty} = \max_{ij} \vert A_{ij}\vert
\end{equation}
\rev{ the $\infty$-norm of $A$.}

Recall that if $A,A' \in \R^{M\times M}$,  then the product matrix $A A'$ has norm
$\| A A'\|_\rev{\infty} \leq M \| A\|_\rev{\infty} \| A'\|_\rev{\infty} $. 
From this it follows that if 
$\| A\|_\rev{\infty} \leq \alpha$, then $\| A^k\|_\rev{\infty} \leq M^{k-1}\alpha^k$.

From the identity
\begin{equation}
(\mathbb{I}+A+A^2+...+A^k)(\mathbb{I}-A) = \mathbb{I}-A^{k-1}
\end{equation}
it follows that the matrix $(1-A)$ is invertible if $\lim_{k\to \infty}A^k=0$, in which case 
\begin{equation}
(\mathbb{I}-A)^{-1} =\lim_{k\to \infty}(\mathbb{I}+A+A^2+...+A^k)\ .
\end{equation}

In particular $(\mathbb{I}-A)$ is invertible if $\| A\|_\rev{\infty} < 1/M$, and one has
\begin{equation}
\| (\mathbb{I} - A)^{-1}\|_\rev{\infty} \leq \frac{1}{1-M\| A\|_\rev{\infty} }\ .
\end{equation}
In particular if $\| A\|_\rev{\infty} < 1/(2M)$, then $\| (\mathbb{I} - A)^{-1}\|_\rev{\infty} < 2$.

Let  $B\in \R^{M}$ be an $M$ component vector. Denote by
\begin{equation}
\| B\|_\rev{\infty} = \max_{i} \vert B_{i}\vert
\end{equation}
its \rev{$\infty$}-norm.
If $A \in \R^{M\times M}$ and  $B\in \R^{M}$, then the vector $AB$ has norm bounded by
\begin{equation}
\| A B\|_\rev{\infty} \leq M \| A\|_\rev{\infty} \| B\|_\rev{\infty} \ .
\end{equation}

As a consequence, if one needs to solve the system of $M$ linear equations for $\lambda_i$, with $A\in \R^{M\times M}$ and $B\in \R^{M}$,
\begin{equation}
(\mathbb{I}-A) \lambda = B\ ,
\label{Eq:Lambda}
\end{equation}
and if $\| A\|_\rev{\infty} < 1/(2M)$, 
then 
\begin{equation}
\max_i \lambda_i \leq 2 M \max_i \vert B_i\vert
\end{equation}

 {\bf System of equations for $\lambda_i$.}

 Eqs. \eqref{Eq:qfg}, viewed as equations for the $\lambda_i$, can be rewritten in the form \eqref{Eq:Lambda}. We will now bound the entry--wise $1$-norm of the corresponding  matrix $A$ and vector $B$.

Explicitly, Eqs. \eqref{Eq:qfg} take the form
\begin{eqnarray}
t_i &=&
\sum_{k\in \bar \Delta_i }f(k;N,s_i) ( t_i+\lambda_i )\nonumber\\
&+& \sum_{k\in \bar \Delta_i }f(k;N,s_i)   \sigma (k-N s_i) 
\nonumber\\
& +&
\sum_{k\in \{0,...,N\}\setminus \bar \Delta_i} f(k;N,s_i) g_n(k) 
\label{Eq:qi3}
\end{eqnarray}
We now consider the contributions of the three terms in Eq. \eqref{Eq:qi3} to $|A_{ij}|$ and $|B_i|$.

{\bf Left hand side and first term.}
The first term on the right hand side in Eq. \eqref{Eq:qi3} is
\begin{equation}
\sum_{k\in \Delta_i }f(k;N,s_i) ( t_i+\lambda_i )
=(1 - \mu_i) ( t_i+\lambda_i )\ .
\label{Eq:1stTerm}
\end{equation}
Therefore, the contribution from Eq. \eqref{Eq:1stTerm} to $|A_{ii}|$ is $\mu_i$. And the joint contribution of Eq. \eqref{Eq:1stTerm}
and the left hand side of Eq. \eqref{Eq:qi3} 
to $|B_i|$ is $\mu_i t_i$.

{\bf Second term.}
\begin{eqnarray}
& & \sigma \vert   \sum_{k\in \Delta_i }f(k;N,s_i)  (k-Ns_i) \vert \nonumber\\
&=&\sigma \vert   \sum_{k\in \{0,...,N\}\setminus \Delta_i} f(k;N,s_i) (k-Ns_i) \vert 
\nonumber\\
& \leq & \mu_i \sigma N
\label{Eq:2ndTerm}
\end{eqnarray}
where we have used Eq. \eqref{Eq:DerBin0}.
Therefore the contribution from Eq. \eqref{Eq:2ndTerm}  to $|B_i|$ is at most $ \mu_i \sigma N.$

{\bf Third term.}
Note that when $\lambda_i=0$ for all $i$, then $g_n(t)\in[0,1]$. Hence the parts of the third term that are independent of the $\lambda$'s are bounded by $\mu_i$.  Therefore the third term contributes at most $\mu_i$ to $|B_i|$.

The parts of the third term proportional to $\lambda$ are given by
\begin{eqnarray}
& &\sum_{\substack{i'=0  \\ {i'\neq i}}}^{M+1} 
\sum_{k\in \bar \Delta_{i'} } f(k;N,s_i)  \lambda_{i'}
\label{eq:thirdterm_1} 
\\
&+&\sum_{i'=0}^M \sum_{ k \in \bar S_{i'}} f(k;N,s_i)\times \nonumber\\
& &
\times(\lambda_{i'+1}-\lambda_{i'} ) \frac{k-N(s_{i'}+\epsilon/3) }{N(s_{i'+1}-s_{i'}-2 \epsilon/3)}\label{eq:thirdterm_2} 
\\
&+&\sum_{i'=0}^M \sum_{ k \in \bar S_{i'}} f(k;N,s_i)  \lambda_{i'} 
\label{eq:thirdterm_3}
\end{eqnarray}

Note that the coefficient of each $\lambda_{i'} $
in lines  \eqref{eq:thirdterm_1} and \eqref{eq:thirdterm_3} is bounded by $\mu$. 
Therefore the contribution from  \eqref{eq:thirdterm_1} to $|A_{ii'}|$ (with $i'\neq i$) is at most $\mu$;
and  the contribution from \eqref{eq:thirdterm_3} to $|A_{ii'}|$ is at most $\mu$. 

Note that $N(s_{i'+1}-s_{i'}-2 \epsilon/3) \geq N \epsilon/3$ and 
that $k-N(s_{i'}+\epsilon/3) \leq N$, hence the factor $(k-N(s_{i'}+\epsilon/3))/(N(s_{i'+1}-s_{i'}-2 \epsilon/3))$ in line \eqref{eq:thirdterm_2} is bounded by $3/\epsilon$.
Hence the coefficient of each $\lambda_{i'}$ and $\lambda_{i'+1}$
in line \eqref{eq:thirdterm_2} is bounded by $3 \mu / \epsilon$. 
Therefore the contribution from \eqref{eq:thirdterm_2} to $|A_{ii'}|$ is at most $6\mu / \epsilon.$

{\bf For large $N$, $\lambda_i$ are unique and small.}

Putting all together, the matrix $A$ in Eq. \eqref{Eq:Lambda} is bounded by
\begin{equation}
\| A\|_\rev{\infty}  \leq \mu (2 + 6/\epsilon) \ .
\end{equation}
And the right hand side of Eq. \eqref{Eq:Lambda} is bounded by
\begin{equation}
\vert B_{i}\vert \leq \mu_i (t_i + \sigma N +1) \leq  \mu (2 + \sigma N)\ .
\end{equation}

Therefore for   large enough  $N$, i.e. $\mu$ sufficiently small, 
$\| A\|_\rev{\infty}  \leq 1/(2M)$. When this is the case 
there is a unique solution for the $\lambda$'s,
and 
\begin{equation}
\vert \lambda_i \vert \leq 2M (2 + \sigma N)\mu \ .
\end{equation}
Since $\mu$ decreases exponentially with $N$, and the other factors increase polynomially with $N$,
by taking $N$ large enough one can ensure that Eq. \eqref{Eq:boundlambda} is satisfied.
\end{proof}

\section{Nondestructive procedures}
\label{Sec:NonDest}

As a first application of Theorem \ref{Thm:SuperStrongAmpl}, we introduce the notion of nondestructive procedure.

\begin{deff}{\bf Nondestructive $(a,b)$--Quantum Verification Procedure.} 
\label{def:NonDest-qvp}
Let $a,b$ be functions as in Definition$~\ref{def:qvpab}$.
An $(a,b)$--quantum verification procedure $Q=\{Q_{n} : n \in \N\}$  is 
{\em nondestructive} if
it outputs both a classical bit (the outcome of the quantum verification procedure) and a quantum state of $m$ qubits (with $m$ the dimension of the witness Hilbert space), such that if the input of the procedure is an eigenstate $\vert \psi_i\rangle$ of $Q_n$ for $x$ then, conditional on the classical output bit being $1$ (i.e. on the procedure accepting), the quantum output is the eigenstate 
$\vert \psi_i\rangle$.
\end{deff}

\begin{thm}
{\bf Properties of nondestructive procedures.} 
\label{Thm:PropNonDest}
Let  $Q=\{Q_{n} : n \in \N\}$ be a nondestructive procedure as in Definition$~\ref{def:NonDest-qvp}$. 
Denote by $\{\vert \psi_i\rangle\}$ the eigenbasis of $Q_n$ for $x$, and by $p_i$ the acceptance probability of $\vert \psi_i\rangle$. Then the following hold:

1) If the quantum input of $Q_n$ is a pure state $\vert \psi^{in} \rangle$ which we express in the eigenbasis of $Q_n$ as 
\begin{equation}
\vert \psi^{in} \rangle=\sum_i \alpha_i \vert \psi_i\rangle\ ,
\label{Eq:psiINQnondest}
\end{equation} 
then, conditional on the classical output bit being $1$, the quantum output will be the pure state 
\begin{equation}
\vert \psi^{out} \rangle = \frac{1}{\sqrt{ \sum_i p_i \vert \alpha_i\vert^2}} \sum_i \sqrt{p_i} \alpha_i \vert \psi_i\rangle\ .
\label{Eq:psiOUTQnondest}
\end{equation}

2) If one uses the state Eq. \eqref{Eq:psiOUTQnondest} as input to the procedure $Q_n$, the probability of acceptance will be larger than the probability of acceptance on the original state Eq. \eqref{Eq:psiINQnondest}:
\begin{equation}
\Pr [Q_{n}(x,\vert \psi^{out} \rangle)=1]\geq \Pr [Q_{n}(x,\vert \psi^{in} \rangle)=1]\ .
\label{Eq:psiINEQQnondest}
\end{equation}
\end{thm}

\begin{proof}
The expression Eq. \eqref{Eq:psiOUTQnondest} is immediate. We prove 
Eq. \eqref{Eq:psiINEQQnondest}. To this end, we introduce the positive operator $Q_x = \sum_i p_i \vert \psi_i \rangle \langle \psi_i \vert$, which is the POVM element corresponding to the classical output of the procedure being $1$ (i.e. accepting). Then we can write \begin{equation}
\Pr [Q_{n}(x,\vert \psi^{in} \rangle)=1] = \langle \psi^{in} \vert Q_x \vert \psi^{in} \rangle
\end{equation}
and
\begin{equation}
\vert \psi^{out} \rangle =  \frac{\sqrt{Q_x} \vert \psi^{in} \rangle}{\sqrt{\langle \psi^{in} \vert Q_x \vert \psi^{in} \rangle }}\ .
\end{equation}

Consequently
\begin{eqnarray}
\Pr [Q_{n}(x,\vert \psi^{out} \rangle)=1] &=& \langle \psi^{out} \vert Q_x \vert \psi^{out} \rangle \nonumber\\
&=& \frac{\langle \psi^{in} \vert Q_x^2 \vert \psi^{in} \rangle}{\langle \psi^{in} \vert Q_x \vert \psi^{in} \rangle}
\nonumber\\
&\geq & 
 \frac{\langle \psi^{in} \vert Q_x \vert \psi^{in} \rangle \langle \psi^{in} \vert  Q_x  \vert \psi^{in} \rangle}{\langle \psi^{in} \vert Q_x \vert \psi^{in} \rangle}
  \nonumber\\
&=& \langle \psi^{in} \vert Q_x \vert \psi^{in} \rangle\ ,
\end{eqnarray}
where we have used the fact that $ I_m \geq \vert \psi ^{in}\rangle \langle \psi^{in}\vert $ with $I_m$ the identity operator.
\end{proof}

We now show that without loss of generality we can take quantum verification procedures to be nondestructive.

\begin{thm}
{\bf Existence of nondestructive procedures.} 
\label{Thm:ExistNonDest}
Let $Q=\{Q_n\}$ be a quantum verification procedures, and let $S=\{S_n\}$,
$S_n = (s_0,s_1,...,s_M,s_{M+1})$,
$T=\{T_n\}$,
$T_n = (t_0,t_1,...,t_M,t_{M+1})$
 be as in Theorem$~\ref{Thm:SuperStrongAmpl}$.
Then there exists a nondestructive quantum verification procedure $Q^{ND}$ such that $Q$ e--maps to $Q^{ND}$ and the polynomial time  computable strictly increasing functions $f_n$ that define the e--map satisfy $f_n(s_i)=t_i$ for all $s_i\in S_n$.
\end{thm}

\begin{proof}

{\bf Step 1: procedure $Q^{(1)}$.}

Use Theorem$~\ref{Thm:SuperStrongAmpl}$ to construct 
 a procedure $Q^{(1)}$ with the property that the polynomial time  computable strictly increasing functions $f_n^{(1)}$ that define the e--map satisfy $ f_n^{(1)}(s_i) = \sqrt{t_i} $.

 {\bf Step 2: procedure $Q^{ND}$.}

 Procedure $Q^{ND}$ is obtained as follows:
 Run the iterative procedure of  Definition \ref{Deff:GAmpProc} with parameter $N=2$ on procedure $Q^{(1)}$, obtaining $2$ bits $z_1$ and $z_2$.
  Accept if $z_1=z_2=1$. Otherwise reject.

 {\bf Step 3: proof that $Q^{ND}$ is a nondestructive procedure with the desired properties.}
\begin{enumerate}
 \item $Q$ e--maps to  $Q^{(1)}$ and $Q^{(1)}$  e--maps to $Q^{ND}$, hence $Q$ e--maps  to $Q^{ND}$. Therefore the eigenbasis $\{\vert \psi_i\rangle\}$ of $Q$ for $x$ is also an eigenbasis 
  $Q^{(1)}$  and of $Q^{ND}$ for $x$.
\item
The polynomial time computable strictly increasing functions $f_n$ that defines the e--map satisfies $f_n(p)=\left( f_n^{(1)} (p)\right)^2$, hence it satisfies $f_n(s_i)=t_i$.
\item If $Q^{ND}$ accepts, then the procedure has ended with a measurement of  $\{ \Pi_0, 1 - \Pi_0\}$ with outcome $1$. Therefore, by property \ref{QNDProperty2} of Theorem \ref{Thm:PropertiesGenAmpl},  if the original state of the witness space was an eigenstate $\vert \psi_i\rangle$ of $Q$, then conditional on $Q^{ND}$ accepting, the state of the witness space is $\vert \psi_i\rangle$.
\end{enumerate}

\end{proof}

\section{Equivalent definitions of $\QMA \cap \coQMA$}
\label{Sec:EquivQMAcapcoQMA}

Definition \ref{Deff:QMACAPcoQMA} is our starting point for studying $\QMA\cap \coQMA$.  It is however not very satisfactory  for defining functional $\QMA\cap\coQMA$. Indeed, recall that functional $\QMA$ is based on the existence of a quantum verification procedure, of witnesses (i.e. of states that are accepted with high probability by the quantum verification procedure and certify that $x\in L$), and of the eigenbasis of the quantum verification procedure.  However it does not seem possible to introduce these notions starting from Definition \ref{Deff:QMACAPcoQMA}. 
The difficulty stems from the fact that the two quantum verification procedures $Q$ and $Q'$ do not commute. 

To circumvent this we introduce two alternative definitions of $\QMA\cap\coQMA$, the first is based on a 3-outcome quantum verification procedure, and the second is based on a $2$-outcome quantum verification procedure. The later is particularly useful, as it allows us to apply the notions of eigenbasis, spectrum, eigenspaces, and amplification in the context of $\QMA\cap\coQMA$.

The inspiration for the following definition comes from the fact that a quantum verification procedure for a language $L\in \QMA$ has two outcomes, but these two outcomes play different roles. Outcome $1$ certifies that $x\in L$ (up to some uncertainty, since outcomes are probabilistic), while  outcome $0$ does not provide information whether $x\in L$ or $x\notin L$. In the case of $\QMA\cap\coQMA$ we need one outcome to certify that $x\in L$, one outcome to certify that $x \notin L$, and one outcome that does not provide information.

\begin{deff}{\bf $(a_3,b_3; a'_3,b'_3)$--Three Outcome Quantum Verification Procedure.} 
\label{def:QMAcoQMA3}
Let $a_3,b_3$ and $a'_3,b'_3$ be  functions as in Definition$~\ref{def:qvpab}$. An $(a_3,b_3; a'_3,b'_3)$--Three Outcome Quantum Verification Procedure is a $3$-outcome procedure $Q^3=\{Q^3_n : n\in \N \}$ whose outcomes are denoted $\{0,L,\overline L\}$, and  such that
 for every $x$ of length $n$,  either both of the following hold:
 \begin{eqnarray}
 \exists \vert \psi \rangle,  \Pr [Q^3_{n}(x,\vert \psi \rangle)=L]\geq a_3,
\label{QMAcoQMA3_1}\\
\forall \vert \psi \rangle, \Pr [Q^3_{n}(x,\vert \psi\rangle )=\overline L]\leq b'_3;
\label{QMAcoQMA3_2}
\end{eqnarray}
or both the following hold:
 \begin{eqnarray}
 \exists \vert \psi \rangle,  \Pr [Q^3_{n}(x,\vert \psi \rangle)=\overline L]\geq a'_3,
\label{QMAcoQMA3_3}\\
\forall \vert \psi \rangle, \Pr [Q^3_{n}(x,\vert \psi\rangle )= L]\leq b_3.
\label{QMAcoQMA3_4}
\end{eqnarray}
\end{deff}

\begin{deff}
\label{Deff:L3}
{\bf The language class $\LTHREE$.} 
Let $a_3,b_3$ and $a'_3,b'_3$ be  functions as in Definition$~\ref{def:qvpab}$. Let $\LTHREE(a_3,b_3; a'_3,b'_3)$ be 
the set of languages  $L\subseteq \{0,1\}^*$ such that there   exists   an $(a_3,b_3; a'_3,b'_3)$--three outcome quantum verification procedure $Q^3=\{Q^3_{n} : n \in \N\}$,
where for every $x$, we have $x\in L$ if and only if Equations \eqref{QMAcoQMA3_1} and \eqref{QMAcoQMA3_2} hold (and consequently we have $x\notin L$ if and only if Equations \eqref{QMAcoQMA3_3} and \eqref{QMAcoQMA3_4} hold).
\end{deff}

The following definitions will enable us to provide a definition of  $\QMA\cap\coQMA$ based on a quantum verification procedure  that only has two outcomes. The idea behind this definition is that if outcome $1$ occurs with high probability this certifies that $x\in L$, if  outcome $0$ occurs with high probability this certifies that $x\notin L$, while if outcomes $1$ and $0$ occur with approximately the same probability, then no information is obtained.

\begin{deff}
{\bf $(a_2,b_2; a'_2,b'_2)$--Quantum Verification Procedure.} 
\label{def:QMAcoQMA4}

Let $a_2,b_2$ and $a'_2,b'_2$ be pairs of functions as in 
Definition$~\ref{def:qvpab}$.

An $(a_2,b_2; a'_2,b'_2)$--Quantum Verification Procedure is
a quantum verification procedure $Q^2=\{Q^2_{n} : n \in \N\}$  such that
 for every $x$ of length $n$,  either both of the following hold:
 \begin{eqnarray}
 \exists \vert \psi \rangle,  \Pr [Q^2_{n}(x,\vert \psi \rangle)=1]\geq \frac{1}{2}+\frac{a_2}{2},
\label{QMAcoQMA4_1}\\
\forall \vert \psi \rangle, \Pr [Q^2_{n}(x,\vert \psi\rangle )=1]\geq \frac{1}{2}-\frac{b'_2}{2};
\label{QMAcoQMA4_2}
\end{eqnarray}
or both the following hold:
 \begin{eqnarray}
 \exists \vert \psi \rangle,  \Pr [Q^2_{n}(x,\vert \psi \rangle)= 0]\geq \frac{1}{2}+\frac{a'_2}{2},
\label{QMAcoQMA4_3}\\
\forall \vert \psi \rangle, \Pr [Q^2_{n}(x,\vert \psi\rangle )= 0]\geq \frac{1}{2}-\frac{b_2}{2}.
\label{QMAcoQMA4_4}
\end{eqnarray}
 \end{deff}

\begin{deff}
\label{Deff:L2}
{\bf The language class $\LTWO$.}  
Let $a_2,b_2$ and $a'_2,b'_2$ be pairs of functions as in Definition$~\ref{def:qvpab}$.
 Let $\LTWO(a_2,b_2; a'_2,b'_2)$ be 
the set of languages the set of languages  $L\subseteq \{0,1\}^*$ such that there   exists   an $(a_2,b_2; a'_2,b'_2)$--quantum verification procedure $Q^2=\{Q^2_{n} : n \in \N\}$,
where for every $x$, we have $x\in L$ if and only if Equations  \eqref{QMAcoQMA4_1} and \eqref{QMAcoQMA4_2} hold (and consequently we have $x\notin L$ if and only if Equations \eqref{QMAcoQMA4_3} and \eqref{QMAcoQMA4_4} hold).
\end{deff}

We note that the language class $\LTWO$ is independent of the bounds used to define it.

\begin{thm} {\bf $\LTWO (a_2, b_2;a'_2, b'_2)$ is independent of  the completeness and soundness probabilities.}
\label{Thm:L2indepbounds}

Let $a_2,b_2$ and ,$a'_2,b'_2$ be pairs of functions as in Definition$~\ref{def:qvpab}$.
For any $r\in \poly$, for any $\tilde a_2,\tilde b_2$, $\tilde a'_2,\tilde b'_2$ pairs of functions as in Definition$~\ref{def:qvpab}$
with $\tilde a_2, \tilde a'_2 < 1-2^{-r}$  and
$\tilde b_2, \tilde b'_2 > 2^{-r}$,
it holds that
$\LTWO(a_2, b_2;a'_2, b'_2) \rev{\subseteq }\LTWO(\tilde a_2, \tilde b_2; \tilde a'_2, \tilde b'_2)$
\end{thm}

 \begin{proof}

\rev{ 
 
{\bf Step 1:  e-map from $Q^2$ to $\tilde Q_2$.}
 
 Let $L$ be a language in $\LTWO(a_2, b_2;a'_2, b'_2)$.
Let $Q^2=\{Q^2_{n} : n \in \N\}$ be the quantum verification procedure 
associated to language $L$, i.e. such that 
 for every $x$, we have $x\in L$ if and only if Equations  \eqref{QMAcoQMA4_1} and \eqref{QMAcoQMA4_2} hold (and consequently we have $x\notin L$ if and only if Equations \eqref{QMAcoQMA4_3} and \eqref{QMAcoQMA4_4} hold).
 
 If we rewrite the four bound Eqs. \eqref{QMAcoQMA4_1}, \eqref{QMAcoQMA4_2}, \eqref{QMAcoQMA4_3},\eqref{QMAcoQMA4_4} in terms of $\Pr [Q^2_{n}(x,\vert \psi \rangle)=1]$, then the corresponding right hand sides, ranked in increasing order are
 \begin{equation}
 0< \frac{1}{2}-\frac{a'_2}{2} < \frac{1}{2}-\frac{b'_2}{2} < \frac{1}{2} +\frac{b_2}{2}
 <\frac{1}{2}+\frac{a_2}{2}<1\ .
 \label{Eq:gapsL2}
 \end{equation}
 
We would like to apply     Theorem \ref{Thm:SuperStrongAmpl}
to the values in Eq. \eqref{Eq:gapsL2}. 
However we cannot do this directly, because some of the differences between two consecutive terms
in Eq. \eqref{Eq:gapsL2} may be exponentially small (in fact the ones corresponding to the first, third, and fifth inequalities).
We note however that language $L$ does not change if we keep the procedure $Q^2$ unchanged, but change the bounds
according to
\begin{eqnarray}
(a_2, b_2;a'_2, b'_2) &\to & ( \bar{a}_2  ,\bar  b_2 ;\bar a'_2 , \bar b'_2)\ ;
\\
\bar{a}_2 &=& {a}_2 -\frac{1}{3q}\ ,\nonumber\\
\bar{b}_2 &=& {b}_2 +\frac{1}{3q}\ ,\nonumber\\
\bar{a}'_2 &=& {a}'_2 -\frac{1}{3q}\ ,\nonumber\\
\bar{b}'_2 &=& {b}'_2 +\frac{1}{3q}\ ,\nonumber
\end{eqnarray}
where $q$ is a polynomial 
such that $a_2-b_2\geq 1/q$ and $a'_2 - b'_2 \geq 1/q$ (which necessarily exists, given the definition of $(a_2, b_2;a'_2, b'_2)$).

We thus apply Theorem \ref{Thm:SuperStrongAmpl} to the procedure $Q_2$ and the
sets
\begin{eqnarray}
S&=&\{S_n :n\in \N\} \nonumber\\
S_n &=& (0,\frac{1}{2}-\frac{\bar a'_2}{2} , \frac{1}{2}-\frac{\bar b'_2}{2} , \frac{1}{2} +\frac{\bar b_2}{2}
 , \frac{1}{2}+\frac{\bar a_2}{2}, 1)\nonumber\\\label{LTWO-S}
\end{eqnarray}
 and
\begin{eqnarray}
T&=&\{T_n :n\in \N\} \nonumber\\
T_n &=& (0,\frac{1}{2}-\frac{\tilde a'_2}{2} , \frac{1}{2}-\frac{\tilde b'_2}{2} , \frac{1}{2} +\frac{\tilde b_2}{2}
 , \frac{1}{2}+\frac{\tilde a_2}{2}, 1)\ ,\nonumber\\\label{LTWO-T}
\end{eqnarray} 
yielding a new procedure $\tilde Q_2$.

{\bf Step 2:  $\LTWO(a_2, b_2;a'_2, b'_2) \subseteq \LTWO(\tilde a_2, \tilde b_2; \tilde a'_2, \tilde b'_2)$.}

We now can show that for the language $L \in \LTWO( a_2,  b_2;  a'_2,  b'_2)$ introduced at the beginning of Step 1, it also holds that
 $L \in \LTWO(\tilde a_2, \tilde b_2; \tilde a'_2, \tilde b'_2)$.

Suppose $x\in L$. 
We  need to show that 
\begin{equation}
 \exists \vert \psi \rangle,  \Pr [\tilde Q^2_{n}(x,\vert \psi \rangle)=1]\geq \frac{1}{2}+\frac{\tilde a_2}{2}
 \label{Eq:L2-AA}
\end{equation}
and
that
\begin{equation}
\forall \vert \psi \rangle, \Pr [\tilde Q^2_{n}(x,\vert \psi\rangle )= 0]\geq \frac{1}{2}-\frac{\tilde b'_2}{2}\ .
\label{Eq:L2-BB}
\end{equation}

We know that
$ \exists \vert \psi \rangle,  \Pr [Q^2_{n}(x,\vert \psi \rangle)=1]\geq \frac{1}{2}+\frac{a_2}{2}$.
Consequently there exists an eigenstate $\vert \psi_i \rangle$ of $Q^2_n$ with acceptance probability
$p_i = \Pr [Q^2_{n}(x,\vert \psi_i \rangle)=1]\geq \frac{1}{2}+\frac{a_2}{2}$.
The state $\vert \psi_i \rangle$  is also an eigenstate of $\tilde Q^2_n$ 
with acceptance probability
$\tilde p_i = \Pr [\tilde Q^2_{n}(x,\vert \psi_i \rangle)=1]\geq \frac{1}{2}+\frac{\tilde a_2}{2}$. Hence Eq. \eqref{Eq:L2-AA} holds.

We also know that
$ \forall \vert \psi \rangle,  \Pr [Q^2_{n}(x,\vert \psi \rangle)=1]\geq \frac{1}{2}-\frac{b'_2}{2}$.
Consequently for all eigenstates $\vert \psi_i \rangle$ of $Q^2_n$, it holds that their  acceptance probability satisfies
$p_i \geq \frac{1}{2} -\frac{b_2}{2}$.
Therefore all eigenstates of $\tilde Q^2_n$ 
have acceptance probability that satisfies
$\tilde p_i  \geq \frac{1}{2} -\frac{\tilde b'_2}{2}$. Hence Eq. \eqref{Eq:L2-BB} holds.

A similar reasoning applies when $x \notin L$.

Therefore $L \in \LTWO (\tilde a_2, \tilde b_2; \tilde a'_2, \tilde b'_2)$ and consequently $ \LTWO(a_2, b_2;a'_2, b'_2) \subseteq  \LTWO (\tilde a_2, \tilde b_2; \tilde a'_2, \tilde b'_2)$.
}
  \end{proof}

We now show that all the above definitions are equivalent.
 
\begin{thm}\label{Thm:L3L2QMAcapcoQMA}
{\bf $\QMA\cap\coQMA = \LTHREE=\LTWO$.} For all $a,b$,  $a,b'$,  pairs of functions as in Definition$~\ref{def:qvpab}$, the following equalities hold
\begin{equation}
\QMA\cap\coQMA = \LTHREE(a,b; a',b') =\LTWO(a,b; a',b') \  .
\label{Eq:L3L2QMAcapcoQMA}
\end{equation}
\end{thm}

\begin{proof}

{\bf Step 1: $\QMA\cap\coQMA (a,b;a',b') \subseteq \LTHREE (a,b;a',b') $.}

Let $a,b$ and $a',b'$  functions as in Definition$~\ref{def:qvpab}$.
Let $L\in \QMA\cap \coQMA(a,b;a',b')$.
Then there
 exist $a,b$ and $a',b'$  functions as in Definition$~\ref{def:qvpab}$, and there exist two quantum verification procedures $Q=\{Q_{n} : n \in \N\}$  and 
$Q'=\{Q'_{n} : n \in \N\}$ such that Equations \eqref{QMAcoQMA1} to \eqref{QMAcoQMA4} 
hold. We now show that $L\in \LTHREE(a,b; a',b')$.

Recall that $Q_n$ takes as input  $(x, \vert \psi \rangle\otimes \vert 0^k \rangle )$ 
where $\vert x \vert = n$, $\vert \psi \rangle$ is a state of
$m$ qubits, and both $m= m(n)$ and $k= k(n)$ are polynomial functions of $n$; and that $Q'_n$ takes as input  $(x, \vert \psi' \rangle\otimes \vert 0^{k'} \rangle )$ 
where $\vert x \vert = n$,  $\vert \psi' \rangle$ is a state of
$m'$ qubits, and both $m'= m'(n)$ and $k'= k'(n)$ are polynomial functions of $n$.

Define $m_3=\max\{m,m'\}+1$, and $k_3=\max\{k,k'\}$. Define the circuit $Q^3_{n}$ which on input
$(x, \vert \tilde \psi \rangle\otimes \vert 0^{k_3} \rangle )$ with $\vert x \vert = n$ and $\vert \tilde \psi \rangle$  a state of
$m_3$ qubits, acts as follows:
\begin{enumerate}
\item Measure the first qubit of $\vert \tilde \psi \rangle$ in the standard basis.
\item If the result of the first measurement is $1$, then carry out the procedure $Q_n$ on input $(x, \vert \tilde \psi \rangle_{2,...,m+1}\otimes \vert 0^{k} \rangle)$ where the second register contains qubits $2,...,m+1$ of  $\vert \tilde \psi \rangle$ and the third register contains the first $k$ ancilla qubits. If the outcome of $Q_n$ is $1$, then output $L$; while if the outcome of $Q_n$ is $0$, then output $0$.
\item
 If the result of the first measurement is $0$, then carry out the procedure $Q'_n$ on input $(x, \vert \tilde \psi \rangle_{2,...,m'+1}\otimes \vert 0^{k'} \rangle)$ where the second register contains qubits $2,...,m'+1$ of  $\vert \tilde \psi \rangle$ and the third register contains the first $k'$ ancilla qubits. If the outcome of $Q'_n$ is $1$, then output $\overline L$; while if the outcome of $Q_n$ is $0$, then output $0$.
 \end{enumerate}
One easily checks that the family of circuits $Q^3=\{Q^3_{n} : n \in \N\}$ thus defined satisfy Eqs. \eqref{QMAcoQMA3_1} and \eqref{QMAcoQMA3_2} when $x\in L$, and satisfy Eqs. \eqref{QMAcoQMA3_3} and \eqref{QMAcoQMA3_4} when $x\notin L$, with $a_3=a, b_3=b, a'_3=a', b'_3=b'$.

{\bf Step 2: $\LTHREE (a,b;a',b')  \subseteq \QMA\cap\coQMA (a,b;a',b')$.} 

Let $a_3,b_3$ and $a'_3,b'_3$  be pairs of functions as in Definition$~\ref{def:qvpab}$. Let  $L\in \LTHREE(a_3,b_3; a'_3,b'_3)$.  
Then there
 exists   a
a three outcome quantum verification procedure $Q^3=\{Q^3_{n} : n \in \N\}$ that satisfies Eqs. \eqref{QMAcoQMA3_1} and \eqref{QMAcoQMA3_2} when $x\in L$, and satisfies Eqs. \eqref{QMAcoQMA3_3} and \eqref{QMAcoQMA3_4} when $x\notin L$. We now show that $L\in \QMA\cap\coQMA(a_3,b_3; a'_3,b'_3)$.

Note that $Q^3_n$ takes as input $(x, \vert \psi \rangle\otimes \vert 0^{k_3} \rangle )$ 
where $\vert x \vert = n$, $\vert \psi \rangle$ is a state of
$m_3$ qubits, and both $m_3= m_3(n)$ and $k_3= k_3(n)$ belong to $\poly$.

Define  two quantum verification procedures $Q=\{Q_{n} : n \in \N\}$  and 
$Q'=\{Q'_{n} : n \in \N\}$ as follows:

$Q_n$ and $Q'_n$ take as input $(x, \vert \psi \rangle\otimes \vert 0^{k_3} \rangle )$, with $\vert \psi \rangle$ a state of
$m_3$ qubits.

On input $(x, \vert \psi \rangle\otimes \vert 0^{k_3} \rangle )$, run $Q^3_n(x, \vert \psi \rangle\otimes \vert 0^{k_3} \rangle )$.

In the case of $Q_n$:\\
if $Q^3_n(x, \vert \psi \rangle\otimes \vert 0^{k_3} \rangle ) = L$ output $1$,\\
if $Q^3_n(x, \vert \psi \rangle\otimes \vert 0^{k_3} \rangle ) = 0$ output $0$,\\
if $Q^3_n(x, \vert \psi \rangle\otimes \vert 0^{k_3} \rangle ) = \overline L$ output $0$.

In the case of $Q'_n$:\\
if $Q^3_n(x, \vert \psi \rangle\otimes \vert 0^{k_3} \rangle ) = L$ output $0$,\\
if $Q^3_n(x, \vert \psi \rangle\otimes \vert 0^{k_3} \rangle ) = 0$ output $0$,\\
if $Q^3_n(x, \vert \psi \rangle\otimes \vert 0^{k_3} \rangle ) = \overline L$ output $1$.

One easily checks that the family of circuits $Q=\{Q_{n} : n \in \N\}$ and $Q'=\{Q'_{n} : n \in \N\}$  thus defined satisfy
satisfy Eqs. \eqref{QMAcoQMA1} and \eqref{QMAcoQMA2} when $x\in L$, and satisfy Eqs. \eqref{QMAcoQMA3} and \eqref{QMAcoQMA4} when $x\notin L$,
 with $a=a_3, b=b'_3$ and $a'=a'_3, b'=b_3$.
 
  {\bf Summary of steps 1 and 2.} 

Since $\QMA \cap \coQMA(a,b; a',b')$ is independent of the bounds $a,b;a',b'$, see Definition \ref{Deff:QMACAPCOQMA}, we have proven the first equality in Eq. \eqref{Eq:L3L2QMAcapcoQMA}, including the fact that $\LTHREE(a,b; a',b')$ is independent of the bounds $a,b;a',b'$.

 {\bf Step 3: $\LTHREE  \subseteq \LTWO$.}

Let $L\in \LTHREE(3/4,1/4; 3/4,1/4)$, and let $Q^3=\{Q^3_{n} : n \in \N\}$, $n$, $m_3$, $k_3$ be defined as in the first two paragraphs of Step 2. We  show that $L\in \LTWO(1/2,1/4; 1/2,1/4)$.

Define  the quantum verification procedures $Q^2=\{Q^2_{n} : n \in \N\}$  as follows:

Run $Q^3_n(x, \vert \psi \rangle\otimes \vert 0^{k_3} \rangle )$ and
\begin{itemize}
\item if $Q^3_n(x, \vert \psi \rangle) = L$ then output $1$;
\item
if $Q^3_n(x, \vert \psi \rangle) = 0$, then  output a random bit drawn uniformly at random from the set $\{0,1\}$;
\item
if $Q^3_n(x, \vert \psi \rangle) = \overline L$  then output $0$.
\end{itemize}

Let us consider the case $x\in  L$.  For brevity in what follows we omit the arguments of $Q^3$. For all input states $\vert \psi \rangle$ we have
\begin{eqnarray}
\Pr[Q^2=1]&=& \Pr[Q^3=L]  + \frac{1}{2} \Pr[Q^3=0]\nonumber\\
&=&\Pr[Q^3=L] \nonumber\\
& &  + \frac{1}{2} \left( 1- \Pr[Q^3=L] - \Pr[Q^3=\overline L] \right)\nonumber\\
&=& \frac{1}{2} +  \frac{1}{2}\Pr[Q^3=L]  -  \frac{1}{2} \Pr[Q^3=\overline L]\ . \nonumber\\
\label{Eq:PQ2=1}
\end{eqnarray}

 Since $x\in L$, we have that for all $\vert \psi \rangle$,
$\Pr[Q^3=\overline L] \leq 1/4$. Furthermore we trivially have  $\Pr[Q^3=L]\geq 0$. Therefore Eq. \eqref{Eq:PQ2=1} implies that for all $\vert \psi \rangle$ the following holds:
\begin{equation}
\Pr[Q^2=1] \geq \frac{1}{2}  -  \frac{1}{2} \frac{1}{4} =  \frac{3}{8} \ .
\end{equation}

Since $x\in  L$, there exists a $\vert \psi \rangle$ such that 
 $\Pr[Q^3=L]\geq 3/4$. 
Therefore Eq. \eqref{Eq:PQ2=1} implies that for this $\vert \psi \rangle$ we have 
\begin{equation}
\Pr[Q^2=1] \geq \frac{1}{2}  +  \frac{1}{2} \frac{3}{4}  -  \frac{1}{2} \frac{1}{4} =  \frac{3}{4} \ .
\end{equation}

Therefore  the family of circuits $Q^2=\{Q^2_{n} : n \in \N\}$   
satisfy Eqs. \eqref{QMAcoQMA4_1} and \eqref{QMAcoQMA4_2} when $x\in L$
 with $a_2=1/2$ and $b'_2=1/4$.  
 
 Similar reasoning shows that when  $x\notin L$, the family of circuits 
 $Q^2=\{Q^2_{n} : n \in \N\}$   
satisfy Eqs. \eqref{QMAcoQMA4_3} and \eqref{QMAcoQMA4_4} 
 with  $a'_2=1/2$ and $b_2=1/4$.

As a consequence
\begin{equation}
 \LTHREE(3/4,1/4; 3/4,1/4) \subseteq  \LTWO(1/2,1/4; 1/2,1/4)\ .
 \end{equation}

 {\bf Step 4: $\LTWO\subseteq \LTHREE$.}  

Let $a_2, b_2$ and $a'_2, b'_2$ be  functions as in Definition$~\ref{def:qvpab}$.
 Let  $L\in \LTWO(a_2,b_2; a'_2,b'_2)$. 
Then there exists a $(a_2,b_2; a'_2,b'_2)$--procedure $Q^2=\{Q^2_{n} : n \in \N\}$ that satisfies Eqs. \eqref{QMAcoQMA4_1} and \eqref{QMAcoQMA4_2} when $x\in L$, and satisfies Eqs. \eqref{QMAcoQMA4_3} and \eqref{QMAcoQMA4_4} when $x\notin L$. 
We now show that 
\begin{equation}
L\in \LTHREE\left((\frac{1}{2}+\frac{a_2}{2})^2, (\frac{1}{2}+\frac{b_2}{2})^2, (\frac{1}{2}+\frac{a'_2}{2})^2, (\frac{1}{2}+\frac{b'_2}{2})^2\right)\ .
\end{equation}

Note that $Q^2_n$ takes as input $(x, \vert \psi \rangle\otimes \vert 0^{k_2} \rangle )$ 
where $\vert x \vert = n$, $\vert \psi \rangle$ is a state of
$m_2$ qubits, and both $m_2= m_2(n)$ and $k_2= k_2(n)$ are in $\poly$. 

Note that since $Q^2$ is a quantum verification procedure, there exists an
eigenbasis of $Q^2$ for $x$, denoted
 $B_{Q^2}(x)=\{\vert \psi^2_i\rangle : 1 \leq i \leq 2^{m_2}\}$. We denote by 
 $p_i= \Pr [Q^2_{n}(x,\vert \psi^2_i \rangle)
=1]$ the acceptance probability of $\vert \psi^2_i\rangle$.

Define  the polynomial time uniform family of quantum circuits $Q^3=\{Q^3_{n} : n \in \N\}$  as follows:

$Q^3_n$ takes as input $(x, \vert \psi \rangle\otimes \vert 0^{k_3} \rangle )$ where $\vert x \vert = n$, $\vert \psi \rangle$ is a state of
$m_2$ qubits, and  $k_3= k_3(n)$ is a polynomial function of $n$.

On input $(x, \vert \psi \rangle)$, carry out a slight modification of the 
Iterative Procedure  of Definition \ref{Deff:GAmpProc} (see note below  Definition \ref{Deff:GAmpProc}) applied to 
$Q^2$ as follows.

The parameter $N$ takes the value $N=2$. 

The functions $g_{\vert x \vert}=g$ are taken as probability distributions over the alphabet $\{L,0,\overline L\}$. They are independent of $\vert x \vert$, and take the deterministic form:
\begin{eqnarray}
&g :\{0,1,2\}\rightarrow \{L,0,\overline L\} &\nonumber\\
&g(2)= L\ ,\ 
g(1)= 0\ ,\ 
g(0)= \overline L&
\end{eqnarray}

Let us bound the probabilities of the outcomes $L$, $0$, and $\overline L$. We can restrict our discussion to the eigenstates of $Q^2$ (because of the absence of interferences, see Equation \eqref{Eq:eigenbasis}).

One finds 
\begin{eqnarray}
\Pr [Q^3_{n}(x,\vert \psi^2_i\rangle )=L]&=& p_i^2\ ,\nonumber\\
\Pr [Q^3_{n}(x,\vert \psi^2_i\rangle )= 0]&=& 2 p_i (1-p_i)\ ,\nonumber\\
\Pr [Q^3_{n}(x,\vert \psi^2_i\rangle )=\overline L]&=& (1-p_i)^2\ .
\end{eqnarray}
 
 One then easily checks that the family of circuits $Q^3=\{Q^3_{n} : n \in \N\}$   thus defined 
satisfy Eqs. \eqref{QMAcoQMA3_1} and \eqref{QMAcoQMA3_2} when $x\in L$, and satisfy Eqs. \eqref{QMAcoQMA3_3} and \eqref{QMAcoQMA3_4} when $x\notin L$,
 with $a_3=(\frac{1}{2}+\frac{a_2}{2})^2$, $b_3=(\frac{1}{2}+\frac{b_2}{2})^2$  and $a'_3=(\frac{1}{2}+\frac{a'_2}{2})^2$, $b'_3=(\frac{1}{2}+\frac{b'_2}{2})^2$.

 {\bf Summary of Steps 3 and 4.}
 
 Since $\LTWO(a,b;a',b')$ and $\LTHREE(a,b;a',b')$   are independent of the bounds $a,b;a',b'$  
 \rev{
 (see Theorem \ref{Thm:L2indepbounds} and remark in paragraph ``Summary of steps 1 and 2''),} 
 steps 3 and 4 show that 
$\LTWO=\LTHREE$. Combining with  steps 1 and 2, we have the second equality in Eq. \eqref{Eq:L3L2QMAcapcoQMA}.

 \end{proof}

 \section{Functional $\QMA \cap \coQMA$.}
 \label{Sec:FQMAcoQMA}
 
\subsection{Definitions} 
 
 We first generalise Definition \ref{Deff:Accepting_and_rejecting_subspaces} as follows:
 
\begin{deff}
\label{Deff:AccRejStates2}
{\bf Accepting and rejecting subspaces.} 
Let $Q=\{Q_n\}$ be a quantum verification procedure and fix $a,a'\in[0,1]$.

We define the following binary relations 
over binary strings and quantum states:
\begin{multline*}
\HH_{Q}^{[0,\frac{1-a'}{2}]\cup[\frac{1+a}{2}, 1]}
(x, \vert \psi \rangle) = 1 \\ \text{ if }
\vert \psi \rangle \in \SPAN ( \HH_{Q}^{\geq  \frac{1+a}{2}}(x) \cup
\HH_{Q}^{\leq  \frac{1-a'}{2}}(x)
 )
 \end{multline*}
 and
\begin{multline*}
\HH_{Q}^{[ \frac{1-a'}{2} , \frac{1+a}{2}]}(x, \vert \psi \rangle) = 1 \\ \text{ if } \vert \psi \rangle \in \SPAN ( \{ \HH_{Q}(x,p) :  \frac{1-a'}{2} \leq p \leq  \frac{1+a}{2} \} )\ .
 \end{multline*}
and to simplify notation, we  denote
\begin{eqnarray}
&&\HH_{Q}^{[0,\frac{1-a'}{2}]\cup[\frac{1+a}{2}, 1]}
(x)=\nonumber\\&&
\quad\quad \left\{ \vert \psi\rangle  \ :\    \HH_{Q}^{[0,\frac{1-a'}{2}]\cup[\frac{1+a}{2}, 1]}
(x, \vert \psi \rangle) = 1 \right\}\ ,\nonumber\\
&&\HH_{Q}^{[ \frac{1-a'}{2} , \frac{1+a}{2}]}(x) =
\nonumber\\&&
\quad\quad
\left\{ \vert \psi\rangle  \ :\    \HH_{Q}^{[ \frac{1-a'}{2} , \frac{1+a}{2}]}(x, \vert \psi \rangle) = 1\right\}\ .
\end{eqnarray}
\end{deff}

 We base our definition of functional $\QMA \cap \coQMA$ on $(a_2,b_2; a'_2,b'_2)$--Quantum Verification Procedures  because these have an eigenbasis which allows the definition of accepting and rejecting subspaces.

\begin{deff}
\label{Deff:FQMAcapCOQMA}
{\bf Functional $\QMA \cap \coQMA$  (${\rm F}(\QMA \cap \coQMA)$).}
Let $a_2,b_2$ and $a'_2,b'_2$ be  functions as in Definition$~\ref{def:qvpab}$. 
The {\em class} ${\rm F}(\QMA \cap \coQMA)(a_2,b_2; a'_2,b'_2)$ is the set 
$\{(
\HH_{Q^2}^{[0,\frac{1-a'_2}{2}]\cup[\frac{1+a_2}{2}, 1]}(x) 
, 
\HH_{Q^2}^{[ \frac{1-b'_2}{2} , \frac{1+b_2}{2}]}(x)
)\}$
of pairs of relations,
where $Q^2$ is an $(a_2,b_2; a'_2,b'_2)$--quantum verification
procedure. 
\end{deff}

We already showed in Theorem \ref{Thm:L2indepbounds} that the language $\LTWO$ is independent of the completeness and soundness bounds used to define it. We now show, using the same argument, that ${\rm F}(\QMA \cap \coQMA)$ is also independent of  the completeness and soundness bounds.

\begin{thm}
{\bf ${\rm F}(\QMA \cap \coQMA)$ is independent of the bounds $(a_2,b_2; a'_2,b'_2)$.} 
\label{Thm:FQMAcoQMASR}
Let $Q^2$ be an $(a_2,b_2; a'_2,b'_2)$--quantum verification
procedure, with $a_2,b_2$ and $a'_2,b'_2$ functions as in Definition$~\ref{def:qvpab}$. 

Let $\tilde a_2,\tilde b_2$ and $\tilde a'_2,\tilde b'_2$ be functions as in Definition$~\ref{def:qvpab}$  
with $\tilde a_2, \tilde a'_2 < 1-2^{-r}$ and $\tilde b_2, \tilde b'_2> 2^{-r}$, 
for some $r\in \poly$.

Then there exists an $(\tilde a_2,\tilde b_2; \tilde a'_2,\tilde b'_2)$--quantum verification
procedure  $\tilde Q^2$, such that there exists an e-map from $Q^2$ to $\tilde Q^2$ and 
the strictly increasing functions $\{f_n\}$ that define the e--map  (see Definition \ref{Deff:StrongReduction}) 
satisfy
$f_n(a_2)=\tilde a_2$, $f_n(b_2)=\tilde b_2$,
$f_n(a'_2)=\tilde a'_2$, $f_n(b'_2)=\tilde b'_2$.
Consequently for all $x$
\begin{equation}
\HH_{\tilde Q^2}^{[0,\frac{1-\tilde a'_2}{2}]\cup[\frac{1+\tilde a_2}{2}, 1]}(x) 
=
\HH_{Q^2}^{[0,\frac{1-a'_2}{2}]\cup[\frac{1+a_2}{2}, 1]}(x) 
\end{equation}
and
\begin{equation}
\HH_{\tilde Q^2}^{[ \frac{1-\tilde  b'_2}{2} , \frac{1+\tilde b_2}{2}]}(x)
=
\HH_{Q^2}^{[ \frac{1-b'_2}{2} , \frac{1+b_2}{2}]}(x)\ .
\end{equation}
\end{thm}

  \begin{proof}
\rev{  
The proof of Theorem \ref {Thm:L2indepbounds} also proves Theorem \ref{Thm:FQMAcoQMASR}.
} 
  \end{proof}

As a consequence the precise values of the bounds $(a_2,b_2; a'_2,b'_2)$  are irrelevant. \rev{ We therefore define ${\rm F}(\QMA \cap \coQMA)$ using standard values for $a_2, a'_2$ and $b_2; b'_2$.}

\begin{deff}
\label{Deff:QMACAPcoQMA2/3}
We define the class ${\rm F}(\QMA \cap \coQMA)$ as ${\rm F}(\QMA \cap \coQMA)(2/3, 1/3; 2/3, 1/3).$
\end{deff}

Since the definition of Functional $\QMA\cap\coQMA$ is in terms of Definition \ref{def:QMAcoQMA4}, what about Functional $\QMA\cap\coQMA$ if one uses the original Definition \ref{Deff:QMACAPcoQMA}?
The proof of 
Theorem \ref{Thm:L3L2QMAcapcoQMA} provides a natural mapping from the pair of quantum verification procedures $Q$ and $Q'$ used in Definition \ref{Deff:QMACAPcoQMA} to the single quantum verification procedure $Q^2$ used in Definition \ref{def:QMAcoQMA4}. Using this mapping  leads to the following definition.

\begin{deff}
\label{Deff:QMACanonical}
{\bf Canonical Definition of Functional $\QMA\cap\coQMA$.}
 Let $a,b$ and $a',b'$ be pairs of functions as in Definition$~\ref{def:qvpab}$. Let $Q$ and $Q'$ be two quantum verification procedures as in Definition \ref{Deff:QMACAPcoQMA}. These procedures define a language $L\in \QMA\cap\coQMA$.
 
 Let $\vert \psi_i\rangle$ be an eigenbasis of $Q$ for $x$, and $p_i$ the corresponding acceptance probability;
 and let $\vert \psi'_i\rangle$ be an eigenbasis of $Q'$ for $x$, and $p'_i$ the corresponding acceptance probability.
 
For all $x$, denote by
 \begin{eqnarray}
  \HH^+(x)
  &=&
  \{ \vert 0 \rangle\otimes \vert \psi_i\rangle : p_i >a\}\nonumber\\ & &
  \cup
   \{ \vert 1 \rangle\otimes \vert \psi'_i\rangle : p'_i > a'\}
 \end{eqnarray} 
 and by
 \begin{eqnarray}
  \HH^-(x)
  &=&
  \{ \vert 0 \rangle\otimes \vert \psi_i\rangle : p_i <b\}\nonumber\\ & &
  \cup
   \{ \vert 1 \rangle\otimes \vert \psi'_i\rangle : p'_i <b'\}
 \end{eqnarray} 

Then the pair of relations 
\begin{equation}
( \HH^+(x),   \HH^-(x)) \in {\rm F}(\QMA \cap \coQMA)(a,b; a',b')
\end{equation}
are the canonical relations associated to the pair of procedures $Q$ and $Q'$ with completeness and soundness probabilities $a,b$ and $a',b'$.
 \end{deff}

From this definition, we see that one of the original problems we were confronted with, namely that $Q$ and $Q'$ do not commute has been circumvented by appending to the witnesses a single qubit whose state, $\vert 0\rangle$ or $\vert 1\rangle$, indicates whether this is a witness for $Q$ or for $Q'$.

\section{Total Functional $\QMA$ equals Functional $\QMA \cap \coQMA $}
\label{Sec:TFQMAequal}

We now prove that Total Functional $\QMA$ equals Functional $\QMA \cap \coQMA $. 

\begin{thm}
{ \bf  ${\rm F}(\QMA \cap \coQMA) = \TFQMA $ }
\label{ThmFQMAcoQMA=TFQMA}
\end{thm}

\begin{proof}
{\bf Step 1: ${\rm F}(\QMA \cap \coQMA)\subseteq \TFQMA $}

Let $Q^2=\{Q^2_n : n \in \N \}$ be an $(2/3,1/3;2/3,1/3)$--quantum verification procedure. For brevity we denote in what follows $a=2/3$ and $b=1/3$.

The corresponding accepting and rejecting subspaces are
$
\HH_{Q^2}^{[0,\frac{1-a}{2}]\cup[\frac{1+a}{2}, 1]}(x)$ and 
$\HH_{Q^2}^{[ \frac{1-b}{2} , \frac{1+b}{2}]}(x)
$. 
\rev{
Note that 
$
\HH_{Q^2}^{[0,\frac{1-a}{2}]\cup[\frac{1+a}{2}, 1]}(x)$ is nonempty for all $x$.
}
Denote by
 $B_{Q^2}(x)=\{\vert \psi^2_i\rangle$ the eigenbasis of $Q^2$ for $x$ and by 
 $p_i= \Pr [Q^2_{n}(x,\vert \psi^2_i \rangle)
=1]$ the acceptance probability of $\vert \psi^2_i\rangle$.

We will  show that there exists a $\frac{1+a^2}{2}$-total quantum verification procedure $Q^T=\{ Q^T_n :  n \in \N \}$
 such that 
\begin{equation}
\HH_{Q^T}^{\geq \frac{1+a^2}{2}}(x) = \HH_{Q^2}^{[0,\frac{1-a}{2}]\cup[\frac{1+a}{2}, 1]}(x)
\label{eq:HHHH1-B}
\end{equation}
and 
\begin{equation}
\HH_{Q^T}^{\leq \frac{1+ b^2}{2}}(x) = \HH_{Q^2}^{[ \frac{1-b}{2} , \frac{1+b}{2}]}(x)
 \ .\label{eq:HHHH2-B}
\end{equation}

We define $Q^T$ as follows:

On input $(x, \vert \psi \rangle)$, run the iterative procedure of Definition \ref{Deff:GAmpProc} on  $Q^2_n$ on input $(x, \vert \psi \rangle)$, for parameter $N=2$, thereby producing a 2 bit output $z_1,z_2$ and functions $G=\{g_{\vert x\vert}\}$ independent of $\vert x\vert$ given by:
\begin{eqnarray}
&g(0)=g(2)=1&\nonumber\\
&g(1)=0&\ .\label{Eq:77}
\end{eqnarray}

One easily checks that 
\begin{eqnarray}
\Pr [Q^T_{n}(x,\vert \psi^2_i\rangle )= 1]&=& p_i^2+ (1-p_i)^2
\label{Eq:78}
\\
\Pr [Q^T_{n}(x,\vert \psi^2_i\rangle )= 0]&=& 2 p_i (1-p_i)\ .
\label{Eq:79}
\end{eqnarray}
\rev{
Equations \eqref{eq:HHHH1-B} and  \eqref{eq:HHHH2-B} then follow from Eq. \eqref{Eq:78}.}

The completeness and soundness thresholds of procedure $Q^T$ are
$\frac{1+a^2}{2}=\frac{13}{18}$ and $\frac{1+b^2}{2}=\frac{10}{18}$.
They can be changed to $2/3$ and $1/3$, see
 Theorem \ref{Thm:FQMASR}, which proves the result.

{ \bf Step 2: $  \TFQMA  \subseteq
{\rm F}(\QMA \cap \coQMA)
$ }

Let $Q=\{Q_n : n \in \N \}$ be an $a$-total quantum verification procedure.
Choose any $b$ such that $a,b$ is a pair of functions as in Definition$~\ref{def:qvpab}$.
 The corresponding accepting and rejecting subspaces are $\HH_{Q}^{\geq a}(x)  $ and $\HH_{Q}^{\leq b}(x) $. 

Denote by
 $B_{Q}(x)=\{\vert \psi_i\rangle\}$ the eigenbasis of $Q$ for $x$ and by 
 $p_i= \Pr [Q_{n}(x,\vert \psi_i \rangle)
=1]$ the acceptance probability of $\vert \psi_i\rangle$.

We show that there exists  an $(a,b;a,b)$--quantum verification procedure
$Q^2=\{Q^2_n : n \in \N \}$
 such that 
\begin{equation}
\HH_{Q^2}^{[0,\frac{1-a}{2}]\cup[\frac{1+a}{2}, 1]}(x) =\HH_{Q}^{\geq a}(x) 
\label{eq:HHHHA}
\end{equation}
and such that
\begin{equation}
\HH_{Q^2}^{[ \frac{1-b}{2} , \frac{1+b}{2}]}(x)=\HH_{Q}^{\leq b}(x) 
 \ .\label{eq:HHHHB}
\end{equation}

We define $Q^2(x,\vert \psi\rangle)$  as follows:

Let $c$ be a bit drawn uniformly at random from the distribution $\{0,1\}$.

If $c=0$, accept

If $c=1$, carry out the procedure $Q$ on input $(x,\vert \psi\rangle)$. Output $Q_n(x,\vert \psi\rangle)$.

One easily checks that 
\begin{eqnarray}
\Pr [Q^2_{n}(x,\vert \psi_i\rangle )= 1]&=& \frac{1+p_i}{2}
\end{eqnarray}
from which Equations \eqref{eq:HHHHA} and  \eqref{eq:HHHHB} follow. (Note that we in fact have that $\HH_{Q^2}^{< 1/2}(x)=\emptyset$,
and consequently 
$\HH_{Q^2}^{[0,\frac{1-a}{2}]\cup[\frac{1+a}{2}, 1]}(x) 
=\HH_{Q^2}^{\geq \frac{1+a}{2}}(x)$
and 
$\HH_{Q^2}^{[ \frac{1-b}{2} , \frac{1+b}{2}]}(x)
=\HH_{Q^2}^{[ \frac{1}{2} , \frac{1+b}{2}]}(x)$).

\end{proof}

\section{If $\TFQMA_W \subseteq \FBQP_W $ then $\QMA \cap \coQMA = \BQP$.}
\label{Sec:InclFBQP}

We show that if the functional $\TFQMA$ class is trivial, in the sense that one can always efficiently produce a witness, then $\QMA \cap \coQMA = \BQP$.

\begin{thm}
{ \bf  
If $\TFQMA_W$ is included in $\FBQP_W$ then $\QMA \cap \coQMA$ equals $\BQP$. }
\end{thm}

\begin{proof}
{\bf Step 1: Trivial direction.}
It follows from Definition$~\ref{def:BQP}$ that $\BQP = \coBQP$. Therefore $\BQP 
\subseteq \QMA \cap \coQMA $.

{\bf Step 2 : Procedure $Q^2$.}
We will show that under the condition $\TFQMA_W(a) \subseteq \FBQP_W(a)$, it holds that $\QMA \cap \coQMA \subseteq \BQP$.

Let $L\in \QMA \cap \coQMA$. Then there exists an 
$(2/3,1/3; 2/3,1/3)$--Quantum Verification Procedure  $Q^2=\{Q^2_{n} : n \in \N\}$  such that
 for every $x$ of length $n$,  either both 
Eqs. \eqref{QMAcoQMA4_1} and  \eqref{QMAcoQMA4_2}, or 
 both 
Eqs. \eqref{QMAcoQMA4_3} and  \eqref{QMAcoQMA4_4} hold.

Denote by
 $B_{Q^2}(x)=\{\vert \psi^2_i\rangle\}$ the eigenbasis of $Q^2$ for $x$ and by 
 $p_i= \Pr [Q^2_{n}(x,\vert \psi^2_i \rangle)
=1]$ the acceptance probability of $\vert \psi^2_i\rangle$.

Note that the definition of $(2/3,1/3; 2/3,1/3)$--procedures implies that
\rev{
\begin{eqnarray}
\mbox{if } x\in L&\mbox{ then }&\frac{1}{3} \leq p_i \leq 1\ \forall i \ ,\label{Eq:xinL}\\
& \mbox{ and }&\exists  \ p_i >  \frac{5}{6}\ ; \label{Eq:5/6} \\
\mbox{if }  x\notin L &\mbox{ then }&  \ 0 \leq p_i \leq \frac{2}{3} \ \forall i \ ,\label{Eq:xnotinL}\\
&\mbox{ and }&\exists  \ p_i <  \frac{1}{6}\ . \label{Eq:1/6}
\end{eqnarray}
}

{\bf Step 3: $\frac{13}{18}$-total procedure $Q^T$.}
The next step of the proof is to construct from $Q^2$ a total procedure $Q^T$ which has the   same eigenbasis $B_{Q^2}(x)$ as $Q^2$. This will allow us to use the hypothesis that $\TFQMA_W(a) \subseteq \FBQP_W(a)$.

\rev{
For simplicity we construct $Q^T$ exactly as in Step 1 of the proof of Theorem \ref{ThmFQMAcoQMA=TFQMA} (using the construction in the paragraph surrounding Eqs. \eqref{Eq:77}, \eqref{Eq:78}, \eqref{Eq:79}). This yields a $\frac{13}{18}$-total procedure $Q^T=\{Q^T_{n} : n \in \N\}$ 
which has the same eigenbasis $B_{Q^2}(x)$ as $Q^2$.
We denote by 
$p_i^T= \Pr [Q^T_{n}(x,\vert \psi^2_i \rangle)
=1]$ the acceptance probability of $\vert \psi^2_i\rangle$ by procedure $Q^T$.
We have that $p_i^T= \beta(p_i)$ with
\begin{equation}
\beta(p_i)=p_i^2 + (1-p_i)^2\ ,
\end{equation}
see Eq. \eqref{Eq:78}.
}

The hypothesis $\TFQMA_W(a) \subseteq \FBQP_W(a)$ implies that there exists an efficiently preparable family of density  matrices $\{\rho(x)\}$
such that for all $x$
\begin{eqnarray}
\Pr [Q^T_{n}(x,\rho(x))=1]\geq \beta\left(\frac{5}{6}\right) =\frac{13}{18}\ .\ \label{Eq:extraTFQMA}
\end{eqnarray}

\rev{
However Eq. \eqref{Eq:extraTFQMA}
does not tell us about the overlap of $\rho$ with the accepting space
$\HH_{Q^T}^{\geq 13/18}(x)=\HH_{Q^2}^{[0,1/6]\cup [5/6,1]}(x)$. In fact this overlap could be exponentially small (see Example \ref{Example:No1} for an explanation). But if we consider a slightly large space, for instance $\HH_{Q^2}^{[0,1/4]\cup [3/4,1]}(x)$, then the overlap of $\rho$ with this larger space can be lower bounded. We proceed to compute such a lower bound.

In what follows we take $x\in L$. In this case $\HH_{Q^T}^{\geq 13/18}(x)=\HH_{Q^2}^{[5/6,1]}(x)$ (the accepting probabilities are all larger than $1/3$, see Eq. \eqref{Eq:xinL}), and we will lower bound the overalp of $\rho$ with $\HH_{Q^2}^{ [3/4,1]}(x)$.
} 

We denote $\rho_{ij}=\langle \psi^2_i\vert \rho \vert \psi^2_j \rangle$ the matrix elements of $\rho$ in the eigenbasis of $Q^2$. Because there is no interferences between eigenbasis states, we have
\begin{eqnarray}
\Pr [Q^2_{n}(x,\rho(x))&=&
\sum_i p_i \rho_{ii}\ ,\nonumber\\
\Pr [Q^T_{n}(x,\rho(x))&=&
\sum_i \beta (p_i)\rho_{ii}\ .
\end{eqnarray}

Hence, for $x\in L$, we have
\begin{eqnarray}
\beta\left(\frac{5}{6}\right)&\leq & 
\Pr [Q^T_{n}(x,\rho(x)) =1 ]\nonumber\\
&=&
 \sum_{i : \frac{1}{3} \leq p_i < \frac{3}{4}} \beta (p_i)\rho_{ii}
 +
 \sum_{i : \frac{3}{4}\leq p_i \leq 1} \beta (p_i)\rho_{ii}
\nonumber\\
&\leq & 
 \sum_{i : \frac{1}{3} \leq p_i < \frac{3}{4}} \beta \left(\frac{3}{4}\right)\rho_{ii}
 +
 \sum_{i : \frac{3}{4}\leq p_i \leq 1} \beta (1)\rho_{ii}
\nonumber\\
&=&  \beta\left(\frac{3}{4}\right)
+ \left( 1 - \beta\left(\frac{3}{4}\right) \right) \sum_{i : \frac{3}{4}\leq p_i \leq 1} \rho_{ii}\nonumber
\end{eqnarray}
where we have used Eq. \eqref{Eq:xinL}, $\beta(1)=1$,   and the fact that over the interval $ [1/3,3/4]$, $\beta(p)$ has its maximum at $3/4$: $\max_{p\in [1/3,3/4]} \beta(p)=\beta(3/4)$.

It then follows that 
\begin{equation}
\sum_{i : \frac{3}{4}\leq p_i \leq 1} \rho_{ii} \geq 
\frac{\beta\left(\frac{5}{6}\right) - \beta\left(\frac{3}{4}\right)}{1 - \beta\left(\frac{3}{4}\right)} =\frac{7}{27}\ .
\label{Eq:7/27}
\end{equation}

\rev{
This is the desired lower bound on the overlap of $\rho$ with $\HH_{Q^2}^{ [3/4,1]}(x)$.
} 

{\bf Step 4: Procedure $Q^P$.}

\rev{
Unfortunately Eq. \eqref{Eq:7/27} cannot be used to directly prove that procedure $Q^2$ belongs to \BQP. We therefore introduce an amplified version of $Q^2$ for which Eq. \eqref{Eq:7/27}  will be sufficient.
} 

We use Theorem \ref{Thm:SuperStrongAmpl} which implies that there exists a
quantum verification procedure $Q^P$ such that $Q^2$
e--maps to $Q^P$ and such that the polynomial time  computable strictly increasing functions $\{f_n\}$ that define the reduction satisfy
\begin{eqnarray}
f_n(3/4)&=& \frac{6}{7}\ ,\label{Eq:6/7}\\
f_n(2/3)&=& \frac{1}{7}\ .\label{Eq:1/7}
\end{eqnarray}

\rev{
Note that $Q^P$ has the same eigenbasis $B_{Q^2}(x)$ as $Q^2$ (and thus as $Q^T$).
We denote by 
$p_i^P= \Pr [Q^P_{n}(x,\vert \psi^2_i \rangle)
=1]$ the acceptance probability of $\vert \psi^2_i\rangle$ by procedure $Q^P$. 
We have that $p_i^P = f_n(p_i)$.
} 

We bound the success probability of $Q^P_n (x,\rho(x))$ when $x\in L$ and $\rho(x)$ is the efficiently preparable state that satisfies Eq. \eqref{Eq:extraTFQMA}:
\begin{eqnarray}
\Pr [Q^P_{n}(x,\rho(x))=1] &=& \sum_i  f_n(p_i)  \rho_{ii} \nonumber\\
&\geq & \sum_{i : \frac{3}{4}\leq p_i \leq 1}  f_n(p_i)  \rho_{ii} \nonumber\\
&\geq & \sum_{i : \frac{3}{4}\leq p_i \leq 1} \frac{6}{7} \rho_{ii} \nonumber\\
&\geq & \frac{7}{27} \frac{6}{7} = \frac{2}{9}
\end{eqnarray}
where we have used Eq. \eqref{Eq:7/27}, Eq. \eqref{Eq:6/7}, and the fact that $f_n$ is strictly increasing. 

Let us bound the success probability of $Q^P_n(x,\rho)$ when $x\notin L$ and $\rho$ is an arbitrary state:
\begin{eqnarray}
\Pr [Q^P_{n}(x,\rho)=1] &=& \sum_i  f_n(p_i)  \rho_{ii} \nonumber\\
& =  & \sum_{i : 0 \leq p_i \leq \frac{2}{3}}  f_n(p_i)  \rho_{ii} \nonumber\\
& \leq   & \sum_{i : 0 \leq p_i \leq \frac{2}{3}}  f_n(\frac{2}{3})  \rho_{ii} \nonumber\\
& =  & \frac{1}{7}
\end{eqnarray}
where we have used Eq. \eqref{Eq:xnotinL}, Eq. 
 \eqref{Eq:1/7}, and the fact that $f_n$ is strictly increasing.

Therefore $Q^P$ is a $(2/9, 1/7)$--procedure that defines the same language $L$ as $Q^2$. Furthermore, when $x\in L$ there exists an efficiently preparable witness. Thus $Q^P$
satisfies the conditions of Definition$~\ref{Deff:Lpoly}$, and $L\in \BQP'(\frac{2}{9},\frac{1}{7})=\BQP$.
\end{proof}

\section{If there exists a $\QMA$ complete problem that robustly reduces to a problem in $\TFQMA$, then  $\QMA = \QMA \cap \coQMA $.}\label{Sec:RobReducQMATFQMA}

\begin{thm}
\label{ThmQMA=QMAcoQMA}
{\bf } 
Let $a' : \N \rightarrow [0,1]$ be a polynomially time computable function.
Let $a,b$ be functions as in Definition$~\ref{def:qvpab}$.
Suppose that there exists 1) an $(a,b)$-quantum verification procedure
$Q=\{Q_{n} : n \in \N\}$ for  a $\QMA$-complete language $L$,
and 2) an $a'$--total quantum verification procedure
$Q^{(T)}=\{Q^{(T)}_{n} : n \in \N\}$, and
3)  $\epsilon \in 1 /\poly$ such that there exists a robust reduction from
 $Q$  to $Q^{(T)}$ with parameter $\epsilon\in 1/\poly$, then $\QMA = \QMA \cap \coQMA$.
\end{thm}

\begin{proof}
 {\bf Step 0: idea of the proof.}

We denote by $(f,\Phi)$ the pair of a polynomial time  computable function and an efficiently implementable channel that define the robust reduction of $Q$ to $Q^{(T)}$, and denote by $\epsilon\in 1/\poly$ the parameter of the robust reduction, see  Definition $~\ref{def:Eps-Reduction}$.

We denote by $L$ the language associated to $Q$. 
We denote by $\bar L$ the complement of $L$. We will show that under the hypothesis of Theorem \ref{ThmQMA=QMAcoQMA},  $\bar L \in \QMA$. To this end we will construct a procedure $Q^{CO}$  which rejects on  all states when $x\in L$, but which accepts on some states when $x\in \bar L$.

The basic idea behind the construction of $Q^{CO}$ is similar to the  proof in the classical case given in \cite{MP91}. Namely, given $x$ and given an input state $\vert \psi \rangle$, we first check whether $\vert \psi \rangle$ is a witness for $Q^{(T)}(f(x))$. If this check is successful, we check if $Q(x, \Phi (\vert \psi \rangle))$ rejects in which case we accept, while if $Q(x, \Phi (\vert \psi \rangle))$ accepts we reject. 

That $Q^{CO}$ thus constructed should have the desired properties follows from the following reasoning. If $x \in L$ then either $\vert \psi \rangle$ is not a witness for $Q^{(T)}(f(x))$ and consequently $Q^{CO}$ rejects, or $\vert \psi \rangle$ is  a witness for $Q^{(T)}(f(x))$, in which case $Q(x, \Phi (\vert \psi \rangle))$ will accept, and therefore $Q^{CO}$ again rejects. On the other hand, if $x \in \bar L$ then the second step (testing whether one has a witness for $Q$) will always fail, and in this case $Q^{CO}$ accepts. Because $Q^{(T)}$ is a total procedure, one can always find a state on which  the first step accepts. Therefore, for $x \in \bar L$ there exists at least one a state on which $Q^{CO}$ accepts.

However there are several complications in the quantum case. We have prepared for these complications  by our earlier theorems and definitions. First it is not obvious that we can use the input state twice. We solve this by replacing $Q^{(T)}$ by a nondestructive version of $Q^{(T)}$ (Theorem \ref{Thm:ExistNonDest}).

Second, we need to be able to amplify the completeness and soundness probabilities without distorting the structure of the witness states. This is achieved through e-maps (Theorems \ref{Thm:SuperStrongAmpl} and \ref{Thm:ExistNonDest}). \rev{The amplified nondestructive version of  $Q^{(T)}$ is denoted $\tilde Q^{(T)}$. We also introduce an amplified version of $Q$, which we denote $\tilde Q$.} 

Thirdly, we need to deal with the eigenstates of $Q^{(T)}$ which have acceptance probability equal to $a'-\delta$, with $\delta$ exponentially small (see Example \ref{Example:No1}). This is addressed by requiring in the statement of the Theorem that the reduction from $Q$ to $Q^{(T)}$ is a robust reduction (Definition \ref{def:Eps-Reduction}). This ensures that eigenstates with acceptance probability equal to $a'-\delta$ (if they exist) can be treated in the same way as the real witnesses (i.e. the eigenstates that have acceptance probability greater or equal to $a'$).

 {\bf Step 1: preliminary definitions.}

Procedure $Q_n$ has as input a classical bit string of length $n$, a witness state of $m(n)$ qubits, and ancilla state of $k(n)$ qubits. 
We denote
\begin{equation}
\eta = 2^{-m-2}\ .
\end{equation}

Using Theorem$~\ref{Thm:ExistNonDest}$ we construct a 
$(1-\eta)$--total nondestructive procedure $\tilde Q^{(T)}$ 
such that $Q^{(T)}$ e--maps to $\tilde Q^{(T)}$, and such that
 the polynomial time  computable strictly increasing functions $\tilde f^{(T)}_n$ that define the e--map satisfy $\tilde f^{(T)}_n(a')=1-\eta$, $\tilde f^{(T)}_n(a'-\epsilon)=\eta$.

 Using Theorem$~\ref{Thm:SuperStrongAmpl}$ we construct a  $(1-\eta, \eta)$--procedure 
 $\tilde Q$ such that $Q$ e--maps to $\tilde Q$, and such that
 the polynomial time  computable strictly increasing functions $\tilde f_n$ that define the e--map satisfy $\tilde f_n(a)=1-\eta$, $\tilde f_n(b)=\eta$.
 
  {\bf Step 2: procedure $Q^{CO}$.}
  
  We define procedure  $Q^{CO}$ as follows:
  \begin{enumerate}
  \item Denote by $(x,\vert \psi\rangle)$ the input to $Q^{CO}$.
  \item Run
$\tilde Q^{(T)}$ on input $(f(x),  \vert \psi\rangle)$.
\item
If $\tilde Q^{(T)}(f(x),  \vert \psi\rangle) = 0$ (i.e. $\tilde Q^{(T)}$ rejects), output $0$ (i.e. $Q^{CO}$ rejects).
\item 
If $\tilde Q^{(T)}(f(x),  \vert \psi\rangle) = 1$,
then apply $\Phi$ to the quantum output of $\tilde Q^{(T)}$. Denote by $\vert \tilde \psi\rangle$ the state so obtained. 
\item
Run $\tilde Q$ on input $(x, \vert \tilde \psi\rangle)$. 
\item
If $\tilde Q_n(x,  \vert \tilde \psi\rangle) = 1$ (i.e. $\tilde Q$ accepts), output $0$ (i.e. $Q^{CO}$ rejects).
\item
If $\tilde Q_n(x,  \vert \tilde \psi\rangle) = 0$ (i.e. $\tilde Q$ rejects), output $1$ (i.e. $Q^{CO}$ accepts).
 \end{enumerate}
 
 We now show that 
 $Q^{CO}$ is a 
 $((1-\eta)^2, 1/4)$--procedure, and that
  the language associated to $Q^{CO}$ is $\bar L$, which proves the result.

  {\bf Step 3: Existence of a witness when $x\in \bar L$.}
  
  Take any $x\in \bar L$. 
  Let $\vert \psi_i\rangle$ be an eigenstate of $Q^{(T)}$ with acceptance probability $p_i\geq a'$. Such an eigenstate exists, since $Q^{(T)}$ is an $a'$--total procedure. 
  We will show that $Q^{CO}$ accepts on input $(x,\vert \psi_i\rangle)$ with probability $\geq (1-\eta)^2$.
  
  To this end note that on this input,
  step 4 of procedure $Q^{CO}$ succeeds with probability $\geq 1 -\eta$, and that the quantum output at step 4 is $\vert \psi_i\rangle$, i.e. it is not affected running $\tilde Q^{(T)}$ in step 2.
  This is where the nondestructiveness of $\tilde Q^{(T)}$ is used.

However $x\in \bar L$. Therefore in step 7, $\tilde Q$ outputs $0$ with probability at least $1-\eta$.
Hence the overall probability that $Q^{CO}$  accepts on input $(x, \vert \psi_i\rangle)$ is at least $(1-\eta)^2$.

  {\bf Step 4: soundness probability of  $Q^{CO}$ on eigenstates of  $Q^{(T)}$.}
  
Take any $x\in L$. 
Let $\vert \psi_i\rangle$ be an eigenstate of $Q^{(T)}$ with acceptance probability $p_i$.
Consider the acceptance probability of 
$Q^{CO}$  on input $(x,\vert \psi_i\rangle)$.

If $p_i < a' -\epsilon$, then $Q^{CO}$ rejects with probability at least $1-\eta$ at step 3.

If $p_i \geq  a' -\epsilon$, then $Q^{CO}$ either rejects at step 3, or passes to step 5. In the latter case, the input of $\tilde Q$ is $(x, \Phi(x,\vert \psi_i\rangle))$, where
$\Phi(x,\vert \psi_i\rangle)$ is a witness for $\tilde Q$ for $x$ (see definition of  robust reductions). 
That is $\tilde Q$ will accept on input $(x, \Phi(x,\vert \psi_i\rangle))$ with probability at least $1-\eta$. Hence the 
probability that $Q^{CO}$ rejects at step 6 is $\geq 1 -  \eta$.

Thus on all  eigenstates $\vert \psi_i\rangle$ of  $Q^{(T)}$, $Q^{CO}$ rejects with probability at least $1- \eta$.

Note that this is the step where we use that the reduction from $Q$ to $Q^{(T)}$ is robust with parameter $\epsilon=1/\poly$. If we set $\epsilon$ to zero, then in step 4 we have no control over how the reduction acts on 
eigenstates of $Q^{(T)}$ which have acceptance probability equal to $a'-\delta$ with $\delta$ exponentially small. 

 {\bf Step 5: soundness probability of  $Q^{CO}$ on arbitrary states.}

\rev{ Take any $x\in L$.} 
 Let us now consider the probability that $Q^{CO}$ rejects on an arbitrary input state $\rho$.
The acceptance probability of $Q^{CO}$ can be written as
\begin{eqnarray}
 \Pr [Q^{CO}_{n}(x,\rho )=1]
 &=&
 \Tr (M_1 \rho ) 
\end{eqnarray}
where $M_1$ is the POVM element corresponding to $Q^{CO}$ accepting.
Note that $M_1$ is a positive  matrix of size $2^m \times 2^m$.

 We have established at Step 4 that for all $i$
\begin{equation}
0\leq \langle \psi_i \vert  M_1 \vert \psi_i \rangle \leq \eta\ .
 \end{equation}
 It then follows 
that for all $i,j$
  \begin{equation}
\vert \langle \psi_i \vert M_1  \vert \psi_j \rangle \vert \leq \eta\ .
\label{Eq:M_1ij}
 \end{equation}
(To see this, consider the restriction of $M_1$ to the 2-dimensional space spanned by $\{ \vert \psi_i\rangle , \vert \psi_j \rangle\}$. Restricted to this space, $M_1$ is a $2\times 2$ positive matrix with diagonal elements bounded by $\eta$. The positivity of this $2\times 2$  matrix implies that its determinant is positive, which implies Eq. \eqref{Eq:M_1ij}).
 
  As a consequence $\Tr M_1^2 \leq \eta^2 2^{2m}\leq 2^{-4}$.
 
 Using the Cauchy-Schwarz inequality
  we have
 \begin{eqnarray}
  \Pr [Q^{CO}_{n}(x,\rho )=1]^2 &=&
\vert  \Tr [ M_1 \rho] \vert^2\nonumber\\
&\leq& \Tr [ M_1^2] \Tr [  \rho^2]\nonumber\\
&\leq&  \Tr [ M_1^2] \leq 2^{-4} \ .
 \end{eqnarray}
 Consequently 
 $ \Pr [Q^{CO}_{n}(x,\rho )=1] \leq 1/4$.

\end{proof}

\section{Conclusion}

In the present work we have shown that 
that $\TFQMA = {\rm F}(\QMA \cap \coQMA)$; that
 if $\FBQP = \TFQMA$, then $\BQP = \QMA \cap \coQMA$; and 
 that if there is a $\QMA$ complete problem that strongly reduces to a problem in
$\TFQMA$, then $\QMA \cap \coQMA = \QMA$. 

These results are not very surprising, as they are immediate generalisations of the analog classical results. However they are significantly more complicated to prove than the classical results. 

 In the appendix we sketch how similar results can be obtained for the classical probabilistic classes $\MA$ and $\coMA$. However  since one can always take $\MA$ to have one sided error  \cite{FGM89}, slightly stronger results should hold in this case. We leave the detailed results to further work.
 
 It would be interesting to better understand the inclusions
 \begin{equation}
\BQP \subseteq \QMA \cap \coQMA \subseteq \QMA\ .
\label{Eq:SubseteqQMAcoQMA_B}
\end{equation}
In one direction, exhibiting problems in $\QMA \cap \coQMA$ that  do not seem to be in $\BQP$, and problems in $\QMA$ that do not seem to be in $\QMA \cap \coQMA$, would provide evidence that the inclusions are strict. 

In the statement of Theorem \ref{ThmQMA=QMAcoQMA} we use the notion of robust reduction, see Definition \ref{def:Eps-Reduction}. It may be that the theorem can be proven using the weaker notion of reduction given in 
definition \ref{def:ReductionQ}, but this probably requires new proof techniques. Alternatively, it is possible that the notion of robust reduction is the natural notion. Some evidence for this is given at the end of subsection \ref{subsec:robreduc}.

\appendix{\large\bf{Appendix: $\MA$ and $\coMA$.}}

The questions we ask and answer in the present work concerning $\TFQMA$ and $\QMA \cap \coQMA$ can also be asked about the analog probabilistic classical classes
 $\TFMA$ and $\MA \cap \coMA$.
 
 However in \cite{FGM89} it was shown that one can always take $\MA$ to have one sided error only. As a consequence the results we obtain for $\TFQMA$ and $\QMA \cap \coQMA$ can certainly be strengthened in the case of $\TFMA$ and $\MA \cap \coMA$. Investigating this would require further work, and is not reported in the present paper. 
 
 Nevertheless most of our proof techniques are in fact classical, and can be easily be adapted to the probabilistic classical classes $\TFMA$ and $\MA \cap \coMA$. In the following paragraphs we provide a dictionary allowing to pass from the classical to the quantum case. We also point out what parts of our arguments are classical, and what parts are intrinsically quantum. 
\revv{The aim of this qualitative appendix is to help to the reader with the main text, as the arguments concerning the quantum classes will be easier to follow when comparing with the corresponding probabilistic classical classes.}
 
 \begin{deff}{\bf Probabilistic Verification Procedure.} 
\label{def:pvp}
 A probabilistic verification procedure is a polynomial time uniform family of classical circuits $C=\{C_{n} : n \in \N\}$ 
with
$C_n$ taking as input $(x, y,z )$, 
where 
 $x \in \{0,1\}^n$ is a binary string of length $n$,
$y$ is a binary string of length $m(n)$ with  $m= m(n)$, $z$ is a string of random bits  (identically independently distributed with bias $1/2$) of length $l(n)$ with $l=l(n) \in \poly$.
For fixed value of $z$, the output of the circuit of $C_n$ is a  bit  which we denote by $C_{n}(x,y,z)$. 
We denote by $p_y=\Pr_z [C_{n}(x, y,z)=1]$.
\end{deff}

Using this definition of probabilistic verification procedures, classes $\MA$, $\coMA$, $\MA \cap \coMA$ are readily defined. Thus for instance languages $L \in \MA$ are those for which there exists $C$ such that if $x \in L$ then 
$\exists y$ such that  $p_y \geq 2/3$ and if  $x \not\in L$ then 
$\forall y$ we have  $p_y \leq 1/3$.

For fixed $x$, the quantum analog of the couples $(y,p_y)$ are the eigenstates and their acceptance probabilities $(\vert \psi_i\rangle, p_i)$. Inputing an arbitrary pure state $\vert \psi \rangle$ or mixed state to a quantum verification procedure is analogous in the classical case to inputing a probabilistic distribution over $y$, see Theorem \ref{Thm:BlockStructure}.

The analog of the accepting and rejecting subspaces, see Definition \ref{Deff:Accepting_and_rejecting_subspaces}, are the following sets
\begin{eqnarray}
S_{C}^{\geq a}(x)&=&
\left\{ y  \ :\   \Pr_z [C_{n}(x, y,z)=1] \geq a  \right\}\ ,\ \\
S_{C}^{\leq b}(x) &=& 
\left\{ y  \ :\   \Pr_z [C_{n}(x, y,z)=1] \leq b  \right\}\ .\ 
\end{eqnarray}

The class Functional $\MA$ ($\FMA$) can be defined in terms of these sets, as well as the class Total Functional $\MA$ ($\TFMA$). (Note that functional classes based on witnesses are not relevant in the classical case, as they  would correspond to probabilistic distribution over $y$'s.).

In the classical case, for fixed $x$ and $y$, one can sample repeatedly from the distribution $p_y$. Remarkably in the quantum case this is also possible, even though cloning of quantum states is impossible\cite{MW05}.

For fixed $x$ and $y$, if one wants to modify the acceptance probability $p_y$, a simple and natural procedure is  to sample $N$ times from the distribution $p_y$, yielding a Binomial distribution $B(N,p_y)$. If one obtains $k$ heads, one tosses a new coin with bias $g(k)$. This yields a new acceptance probability $p'_y = P_g( p_y)$ (where the function $P_g$ depends on $N$ and the choice of function $g(k)$).
This classical procedure can be implemented quantumly,  in which case we call it an iterative procedure, see Definition \ref{Deff:GAmpProc}.

 We show in Theorem \ref{Thm:GAmpStrongReduc} that if $g$ is increasing and non constant, the function relating  the new and old acceptance probabilities $p'_y(p_y)$ is strictly increasing. We then show in Theorem \ref{Thm:SuperStrongAmpl} that by taking $N$ sufficiently large, and appropriate $g$, the strictly increasing function $p'_y(p_y)$ can be made to pass through a finite number of points. These two results are purely classical, and apply both the classical and quantum iterative procedures.
 
 The notion of non destructive procedure, Definition \ref{def:NonDest-qvp}, is non trivial in the quantum case. In the classical case it is trivial, as one can copy at will the input.

In section \ref{Sec:EquivQMAcapcoQMA} we give three definitions of $\QMA \cap \coQMA$ and show that they are equivalent. Exactly the same definitions can be given for $\MA \cap \coMA$, and the proofs of equivalence also holds in the classical case. In particular $\MA \cap \coMA$ can be defined in terms of a single probabilistic verification procedure. The class Functional $\MA \cap \coMA$ can then be defined in terms of the sets
\begin{eqnarray}
&& S_{C}^{[0,\frac{1-a'}{2}]\cup[\frac{1+a}{2}, 1]}
(x)=\nonumber\\&&
\quad\quad \left\{ y  \ :\    \Pr_z [C_{n}(x, y,z)=1] \in \frac{1-a'}{2}]\cup[\frac{1+a}{2}, 1]
\right\}\ ,\nonumber\\
&& S_{C}^{[ \frac{1-a'}{2} , \frac{1+a}{2}]}(x) =
\nonumber\\&&
\quad\quad
\left\{ y  \ :\    \Pr_z [C_{n}(x, y,z)=1] \in [ \frac{1-a'}{2} , \frac{1+a}{2}] \right\}\ .
\end{eqnarray}
as in Definition \ref{Deff:FQMAcapCOQMA}.

One can then prove  analogs of the key results of sections  \ref{Sec:TFQMAequal}, \ref{Sec:InclFBQP}, and \ref{Sec:RobReducQMATFQMA}, namely  
 that $\TFMA$ equals Functional $\MA \cap \coMA$, that if $\TFMA$ is included in $\FBPP$ then $\MA \cap \coMA$ equals $\BPP$, and that if there exists a $\MA$ complete problem that robustly reduces to a problem in $\TFMA$, then  $\MA = \MA \cap \coMA $. Indeed the proofs of all these results are in fact classical.

\bibliographystyle{plain}

\end{document}